\documentclass[12pt]{article}

\usepackage{amssymb,amsmath,amsfonts,eurosym,geometry,ulem,graphicx,caption,color,setspace,sectsty,comment,footmisc,caption,pdflscape,subfigure,array,hyperref}

\usepackage[numbers]{natbib}
\usepackage{url}

\usepackage{multicol}
\usepackage{adjustbox}

\usepackage{lipsum}
\usepackage{epsfig}
\usepackage{subfigure}
\usepackage{graphicx}
\usepackage{algorithm}
\usepackage[noend]{algorithmic}
\usepackage{amsmath}
\usepackage[table]{xcolor}

\newenvironment{proof}[1][Proof]{\noindent\textbf{#1.} }{\ \rule{0.5em}{0.5em}}

\newcolumntype{L}[1]{>{\raggedright\let\newline\\arraybackslash\hspace{0pt}}m{#1}}
\newcolumntype{C}[1]{>{\centering\let\newline\\arraybackslash\hspace{0pt}}m{#1}}
\newcolumntype{R}[1]{>{\raggedleft\let\newline\\arraybackslash\hspace{0pt}}m{#1}}

\begin{document}

\begin{titlepage}
\title{When to Invest in Security? Empirical Evidence and a Game-Theoretic Approach for Time-Based Security}
\author{Sadegh Farhang\thanks{College of Information Sciences and Technology, The Pennsylvania State University. Email: farhang@ist.psu.edu} \and Jens Grossklags\thanks{Chair for Cyber Trust, Department of Informatics, Technical University of Munich. Email: jens.grossklags@in.tum.de}}
\date{May 2017}
\maketitle
\begin{abstract}
\noindent \textit{Games of timing} aim to determine the optimal defense against a strategic attacker who has the technical capability to breach a system in a stealthy fashion. Key questions arising are \textit{when} the attack takes place, and \textit{when} a defensive move should be initiated to reset the system resource to a known safe state.

In our work, we study a more complex scenario called Time-Based Security in which we combine three main notions: \textit{protection} time, 
%($\mathbf{p}$), 
\textit{detection} time,
%($\mathbf{d}$), 
and \textit{reaction} time. 
%($\mathbf{r}$). 
Protection time represents the amount of time the attacker needs to execute the attack successfully. In other words, protection time represents the inherent resilience of the system against an attack. Detection time is the required time for the defender to detect that the system is compromised. Reaction time is the required time for the defender to reset the defense mechanisms in order to recreate a safe system state. 

In the first part of the paper, we study the VERIS Community Database (VCDB) and screen other data sources to provide insights into the actual timing of security incidents and responses. While we are able to derive distributions for some of the factors regarding the timing of security breaches, we assess the state-of-the-art regarding the collection of timing-related data as insufficient.

In the second part of the paper, we propose a two-player game which captures the outlined Time-Based Security scenario in which both players move according to a periodic strategy. We carefully develop the resulting payoff functions, and provide theorems and numerical results to help the defender to calculate the best time to reset the defense mechanism by considering protection time, detection time, and reaction time. 
%$\mathbf{p}$, $\mathbf{d}$, and $\mathbf{r}$.\\
%\vspace{0in}\\
%\noindent\textbf{Keywords:} key1, key2, key3\\
%\vspace{0in}\\
%\noindent\textbf{JEL Codes:} key1, key2, key3\\

\bigskip
\end{abstract}
\setcounter{page}{0}
\thispagestyle{empty}
\end{titlepage}
\pagebreak \newpage

\doublespacing

%\section{Introduction} \label{sec:introduction}

%\section{Literature Review} \label{sec:literature}

%\section{Data} \label{sec:data}

%\section{Results} \label{sec:result}

%\section{Discussions} \label{sec:discussion}

%\section{Conclusion} \label{sec:conclusion}

\section{Introduction}
\label{sec:Intro}
On Monday early morning, February 17, 2014, an Ethiopian Airlines co-pilot informed ground control that he had highjacked flight ET-702 from Addis Ababa to Rome and was planning to fly the plane over Swiss air space towards Geneva. The plane eventually landed in Geneva at 6:02am local time. Despite the uncertainty about the motives of the highjacker, and heightened worries about terrorism, no escort could be provided by the Swiss Air Force since they do not operate before 8am on weekdays. They also do not operate during lunch breaks, or on weekends \cite{Agence14}.

We are using this opening real-world example to illustrate different strategic components of a security game unfolding over time. These games of timing help us to understand the defense against a motivated attacker who has the technical capability to breach a system in a \textit{stealthy} fashion. A first key question that then arises is \textit{when} the attack takes place, and \textit{when} a defensive move should be initiated to reset the system resource to a known safe state. Such scenarios have recently been a new focus of study, for example, with the FlipIt game \cite{dijk2013flipit}.

However, the Swiss Air Force example demonstrates that security situations typically involve more complexity which we aim to capture with a model of \textit{Time-Based Security}; a concept which has been informally introduced by Schwartau in 1999 \cite{Schwartau99}. In this model, we combine three main notions: \textit{protection} time ($\mathbf{p}$), \textit{detection}\footnote{We will use \textit{detection} time and \textit{discovery} time interchangeably throughout the paper.} time ($\mathbf{d}$), and \textit{reaction} time ($\mathbf{r}$). Protection time represents the amount of time the attacker needs to execute her attack successfully. In other words, protection time represents the inherent resilience of the system against an attack. Detection time is the required time for the defender to detect that his system has been stealthily compromised. This is an important facet of security decision-making as evidenced by data indicating that organizations on average fail to detect attacks for over 225 days \cite{Nadella15}. Finally, reaction time is the required time for the defender to reset his defense mechanisms in order to recreate a safe system state. 

In the example, the pilot stealthily took ownership of a plane at a particular day and time when he likely had many previous opportunities during his employment with Ethiopian Airlines. He then proceeded to direct the plane to his target destination. The total time to take possession of the plane and to reach the target destination is called the protection time in our model. The pilot informed ground control about the highjacking, and thereby significantly shortened the detection time of the defenders. It also changed a stealthy attack to a known attack. Finally, reaction time was excessive due to the non-responsiveness of the Swiss Air Force.% outside of office hours.

In our work, we first study the VERIS\footnote{VERIS is an acronym for the Vocabulary for Event Recording and Incident Sharing.} Community Database (VCDB)~\cite{vcdb} to shed light on the question of the actual timing of security incidents and responses. Based on our VCDB analysis, we provide the distribution of the detection time for malware activities and hacking incidents. We also provide a selection of heuristics about the protection time and the reaction time based on our VCDB analysis. In addition, we screen further data sources to shed light at the timing of security incidents. Furthermore, motivated by our empirical analysis, we propose a game-theoretic model for the concept of Time-Based Security. Based on our model specifications, we carefully derive the resulting payoff functions, and provide numerical results to help the defender to calculate the best time to reset his defense mechanism and Nash Equilibrium by considering protection time, detection time, and reaction time. 
%According to our analytic and numerical analyses, our findings are as follows: (1) The defender's strategy to check the state of the resource continuously is not always the defender's best response. (2) The defender checks the state of the resource with delay when the attacker moves fast. (3) The defender checks the state of the resource continuously when the attacker moves slower for higher values of $c_{\mathbf{D}}$, $\mathbf{p}$, $\mathbf{d}$, and $\mathbf{r}$. (4) The defender's strategy to check the state of the resource continuously is not necessarily part of the Nash equilibrium of the game. %Our paper is a work-in-progress report of our research on Time-Based Security, and only captures a subset of results given the full mathematical complexity of the scenario, but it demonstrates the overall approach in sufficient detail. 
We anticipate the further development of Time-Based Security to positively impact the study of security games of timing, as well as security practice.

\textbf{Roadmap:} In Section~\ref{sec:Related}, we discuss related work on security games of timing. In Section~\ref{sec:Data}, we provide empirical evidence on the impact of timing for security decision-making by analyzing data from the VERIS Community Database (VCDB) and other data sources. In Section~\ref{sec:SysMod}, we develop our model of time-based security, followed by analytic results in Section~\ref{sec:Ana}. In Section~\ref{sec:Num}, we provide numerical results. We conclude in Section~\ref{sec:Conclu}.

%%%%%%%%%%%%%%%%%%%%%%%%%%%%%%%%%%%%%%%%%%%%%%%%%%%%%%%%%%
\section{Related Work}
\label{sec:Related}
%\subsection{Security Economics and Games of Timing}
%The economics of security decision-making is a rapidly expanding field covering theoretical, applied, and behavioral research. Theoretical work utilizes diverse game-theoretic and decision-theoretic approaches, and addresses abstract as well as applied scenarios. A central research question has been how to optimally determine security investments \cite{Gordon02,varian2004system,ioannidis2009investments,panaousis2014cybersecurity}, e.g., by selecting from different canonical defense actions (i.e., protection, mitigation, risk-transfer) \cite{Johnson11,lelarge2009economic}, and how such investments are influenced by the actions of strategic attackers \cite{Fultz09,Schechter03}. Another frequently addressed aspect has been the consideration of interdependence of security decision-making and the propagation of risks \cite{Chen2011,grossklags2008secure,johnson2014complexity,kunreuther2003interdependent}. Recent surveys summarize these research efforts in great detail \cite{bohme2010modeling,Laszka_survey,Manshaei13}. 
Game theory has been used extensively for studying different aspects of security and privacy~\cite{Manshaei13}. These analyses include, but are not limited to, interdependent security~\cite{grossklags2008secure,Johnson11,Laszka_survey} and privacy~\cite{pu2014economic}, uncertainty about the type of the connected clients to a server~\cite{farhang2014dynamic,liu2006bayesian}, timing of security incidents~\cite{dijk2013flipit}, and many privacy issues such as location privacy~\cite{farhang2015phy,shokri2012protecting}.
One of the critical aspects of securing resources is the \textit{timing} of appropriate actions to successfully thwart attacks. The time-related aspects of security choices have been studied from both theoretical~\cite{Blackwell49,radzik1996results} and empirical perspectives~\cite{grossklags2014task,reitter2013risk,nochenson2013behavioral}. 
However, to the best of our knowledge, there is no extant research specifically investigating the timing of real-world security incidents and its connection to the games of timing literature.

%Decision-making in security has been studied thoroughly an often overlooked but critical decision dimension for successfully securing resources is the consideration of \textit{when} to act to successfully thwart attacks. Scholars have studied such time-related aspects of tactical security choices from both theoretical and empirical viewpoints. From theoretical point of view, they primarily focus on zero-sum games called \textit{games of timing}~\cite{Blackwell49}. The theoretical contributions on some subclasses of these games have been surveyed by \cite{radzik1996results}. 

The optimal timing of security decisions became a lively research topic with the development of the FlipIt game~\cite{bowers2012defending,dijk2013flipit}. 
In FlipIt, two players compete for control of a critical resource to be gaining continuous benefits. Players take ownership of the resource by making costly moves (``flips''). However, these moves are under incomplete information about the current state of possession of the contested resource. In the original FlipIt papers, equilibria and dominant strategies for basic cases of interaction are studied~\cite{bowers2012defending,dijk2013flipit}. Further, extended versions of FlipIt for periodic strategies have been studied by considering the impact of an audit move to check the state of the resource~\cite{pham2012compromised}.

Laszka et al. study FlipIt with non-covert defender moves. They consider both targeting and non-targeting attackers with non-instantaneous moves~\cite{laszka2013mitigation,laszka2013mitigating}. To extend the application of the FlipIt game for multiple contested resources, FlipThem is proposed~\cite{laszka2014flipthem}. FlipThem is also extended to consider the case where the attacker attempts to compromise enough resources to reach a threshold \cite{leslie2015threshold}. Likewise, Pal et al. study a game with multiple defenders~\cite{pal_security}. Farhang and Grossklags~\cite{farhang2016flip} propose a game-theoretic model in the tradition of FlipIt called \textit{FlipLeakage} where a variable degree of information leakage exists that depends on the current quality of defense. In addition, FlipIt has been studied with resource constraints placed on both players \cite{zhang_stealthy}. This line of work has been enriched by a model for dynamic environments with adversaries discovering vulnerabilities according to a given vulnerability discovery process and vulnerabilities having an associated survival lifetime \cite{johnson2015games}. Axelrod and Iliev~\cite{axelrod2014timing} study the timing of cyber-conflict from a different perspective. They proposed a model for the question of when the resource should be used by the attacker to exploit an unknown vulnerability. Note that the focus of their proposed model is on the attacker rather than the defender.

Other researcher teams have studied games with multiple layers in which in addition to external adversaries the actions of insiders (who may trade information to the attacker for a profit) need to be considered \cite{feng_stealthy,hu2015dynamic}. Taking a different perspective, a discrete-time model with multiple, ordered states in which attackers may compromise a server through cumulative acquisition of knowledge (compared to one-shot takeovers) has been proposed \cite{wellman2014empirical}.
Finally, the FlipIt framework has been extended to the context of signaling games to calculate sender's (one of the players in the signaling game) prior beliefs about the receiver's types (i.e., the second player in the signaling game which has two potential types) \cite{pawlick2015flip}.

From an empirical point of view, Liu et al.~\cite{liu2015cloudy} use data from the VERIS Community Database (VCDB) to predict cybersecurity incidents based on externally observable properties of an organization's network. Further, Sarabi et al.~\cite{sarabi2016risky} try to assess a company's risk from security incidents and the distribution of security incidents by using an organization's business details. They also use VCDB data for their analysis. Edwards et al.~\cite{edwards2016hype} investigate trends in data breaches by using Bayesian generalized linear models. In their analysis, they use data from the Privacy Rights Clearinghouse~\cite{PrivacyRights}. To the best of our knowledge, none of these studies focus on a detailed investigation of timing-related aspects.

Our theoretical work differs from the previous games of timing literature. We take into account the Time-Based Security approach in our model by integrating the notions of protection time, detection time, and reaction time to propose a more realistic game-theoretic framework. % for the defender's best time to update his defense mechanism. 
%Second, we differentiate between the ability of the defender to detect attacks or not. In doing so, we propose two separate models for a defender without detection ability as well as a defender with detection ability. Then, we combine these two models to propose more a comprehensive game-theoretic framework for the defender's best time to update his defense mechanism. 
As a result, our work aims to overcome several important simplifications in the previous literature which would limit their applicability to practical defense scenarios. Moreover, in this paper, we focus on the values of \textit{protection time}, \textit{discovery time}, and \textit{reaction time} in different publicly available datasets. To the best of our knowledge, this is the first work taking this angle during the exploration of datasets and based on the analysis to propose a theoretical framework. 
%%%%%%%%%%%%%%%%%%%%%%%%%%%%%%%%%%%%
\section{Security Incident Data about Timing}
\label{sec:Data}
In this section, we discuss empirical data sources which shed light on the question of the actual timing of security incidents and responses.

In the TBS model, we have three main notions: \textit{protection} time ($\mathbf{p}$), \textit{detection} time ($\mathbf{d}$), and \textit{reaction} time ($\mathbf{r}$). Our focus is to determine actual values of these parameters in practice. $\mathbf{p}$ represents the amount of time the attacker needs to execute the attack successfully. In other words, $\mathbf{p}$ is the amount of time that the defender can protect the system against an attack (i.e., the system's inherent resilience). $\mathbf{d}$ is the elapsed time in order for the defender to recognize that the system is compromised. Finally, we can understand the reaction time, $\mathbf{r}$, as the required time for the defender to reset/update the defense mechanism in order to create a safe system state. 

We do not expect that available field data will exactly match our definitions, but we anticipate that existing data sources provide some indication of the magnitude of these parameters.

\subsection{Dataset}
\label{sec:dataset}

A particularly relevant industry report for our work is Verizon's annual Data Breach Investigations Report (DBIR)~\cite{DBIR}. The report from 2016 draws on more than 100,000 incidents (3,141 of them are confirmed data breaches) from different sources including the VERIS Community Database (VCBD)~\cite{vcdb}. In particular, the report shows the percent of breaches for which the time of compromise and the exfiltration time is known,
%They only represent the percent of compromise (or exfiltration) for the following categories: ``Seconds'', ``Minutes'', ``Hours'', ``Days'', ``Weeks'', ``Months'', and ``Years'', 
see Figure~\ref{fig:verizon}. Time to compromise is defined as the time from the beginning of an attack to the first point at which a security attribute of an information asset was compromised. Exfiltration refers to the time from initial compromise to the time when valuable data was taken away from the victim (i.e., the first point in time at which non-public data was taken from the victim environment). Note that in Figure~\ref{fig:verizon}, $n$ shows the number of incidents with the available corresponding timing aspect. 

Unfortunately, Figure~\ref{fig:verizon} which stems from the DBIR report is not complemented with detailed commentary. For example, it is not clear what the sources of incidents are. Moreover, the figure does not provide the exact distribution of compromise time and exfiltration time. However, the given data suggests that protection time is very short, which is highly unfavorable from a defender's perspective. 

%Here, in our analysis, we mainly focus on incidents arose from ``Malware'' activities or ``Hacking'' incidents. Moreover, this figure does not give the exact distribution of compromise time and exfiltration time. However, it suggests that the protection time is very short. 

In the following, we take a closer look at the available data in the VCDB~\cite{vcdb}, which likely constitutes a major part of the data used for Figure~\ref{fig:verizon}, to provide a more realistic perspective about the timeline of security breaches and their availability. 

\begin{figure}
	\centering
	\includegraphics[scale=0.7]{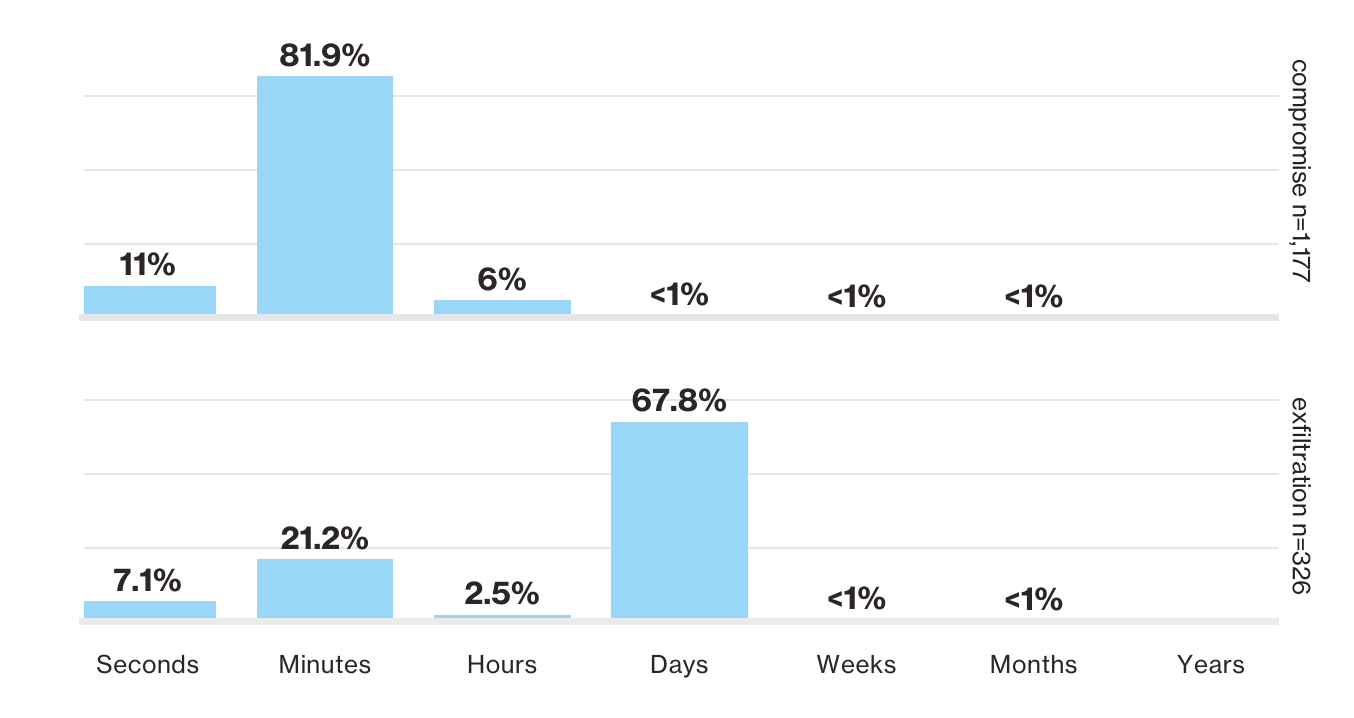}
	\caption{Time to compromise and exfiltration from DBIR report~\cite{DBIR}}
	\label{fig:verizon}
\end{figure}

The VCDB is currently composed of 5856 reports of publicly disclosed data breaches. This dataset includes incidents that occurred up to and including 2016. The structure of entries in the VCDB is based on the Vocabulary for Event Recording and Incident Sharing (VERIS)~\cite{VERIS} which also provides a description on how to report to the VCDB. Each entry of the VCDB includes the following: incident timeline, the victim organization, the actor and motive, the type of incident, how an incident occurred, the impact of an incident, links to news reports or blogs documenting the incident. Note that some of these fields are not mandatory, because a victim organization may not have all the information or may not want to disclose all the details of an incident. In the following, we only focus on the fields that are related to our study, which are \textit{action}, \textit{timeline}, \textit{impact}.

The first field that we take into account is ``action'' that represents information about the type of attack. In VCDB, there are seven primary categories: ``Malware'', ``Hacking'', Social'', ``Misuse'', ``Physical'', ``Error'', and ``Environmental''. Each category has additional fields to provide more information. For example, ``Hacking'' can be the result of SQL injection or a brute force attack. Note that based on the VERIS description, it is possible that some incidents are associated with more than one category. 

In our analysis, we only focus on the incidents resulting from ``Malware'' and/or ``Hacking'', since we are primarily interested in cybersecurity incidents caused by a malicious entity. In VCDB, there are $439$ entries with the ``Malware'' tag and $1655$ ``Hacking'' incidents. Taking into account that many of these incidents have more than one category, there are 1795 distinct incidents with ``Malware'' and/or ``Hacking'' labels, which we consider for further analysis. 

The next field we take into account is ``timeline'' which provides information about the timing of security incidents. The timeline field has five sub-fields: ``incident date'', ``time to compromise'', ``time to exfiltration'', ``time to discovery'', and ``time to containment''. The only mandatory part of ``timeline'' is to provide the year of an incident, while other parts are optional. Note that there may be different approaches among organizations to measure incident dates, and VERIS ``suggests the point of initial compromise as the most appropriate option to as the primary date for time-based analysis and trending of incidents"\''~\cite{VERIS}. 

%The following definitions are provided. 
Furthermore, ``incident date'' represents the first point at which a security attribute of an information asset was compromised. However, it is not obvious whether ``time to compromise'' has a distinct meaning in the VCDB from \textit{incident date}. For the entries with non-empty \textit{time to compromise}, the value of \textit{time to compromise} is equal to \textit{incident date}. ``Time to exfiltration,'' as previously stated, represents the initial compromise to data exfiltration. It is only available for data compromise incidents. Further, initial compromise to incident discovery is represented by the ``time to discovery'' field and the field ``time to containment'' shows the initial compromise to containment or restoration. 

As we mentioned earlier, we only consider those $1795$ entries resulting from ``Malware'' activities and/or ``Hacking'' incidents. $473$ of them have at least one additional non-empty field in addition to the mandatory \textit{incident date} field. As mentioned previously, every entry with non-empty \textit{time to compromise} field has the same value to its corresponding \textit{incident date}. Therefore, the \textit{time to compromise} field does not provide more information about the attack timing. Note that the timeline unit for other fields, i.e., \textit{time to exfiltration}, \textit{time to discovery}, and \textit{time to containment}, are ``NA'', ``Seconds'', ``Minutes'', ``Hours'', ``Days'', ``Weeks'', ``Months'', ``Years'', ``Never'', and ``Unknown''. In the following subsections, we investigate the values of timeline sub-fields to observe how the distribution of our desired parameters, i.e., $\mathbf{p}$, $\mathbf{d}$, and $\mathbf{r}$, are in practice. 

\subsection{Discovery Time}
\label{sub:discover}

With respect to the discovery time, there are $325$ entries with non-empty \textit{time to discovery} and without ``Unknown'' or ``NA'' field value. Some of these $325$ entries just provide the unit of the discovery time, e.g., ``Hours'', without specifying the exact value. We also exclude such entries from our analysis. In total, there are $150$ entries with exact values for the discovery time. In these $150$ entries, the average value of discovery time is equal to $198.2539$ days. The maximum value of discovery time is 6 years and the minimum value is 10 hours. Figure~\ref{fig:DiscoveryTime} represents the distribution of the discovery time in the VCDB given the stated limitations. 

\begin{figure}
	\centering
	\begin{subfigure} [Discovery Time]{ %
			\includegraphics[scale=0.35]{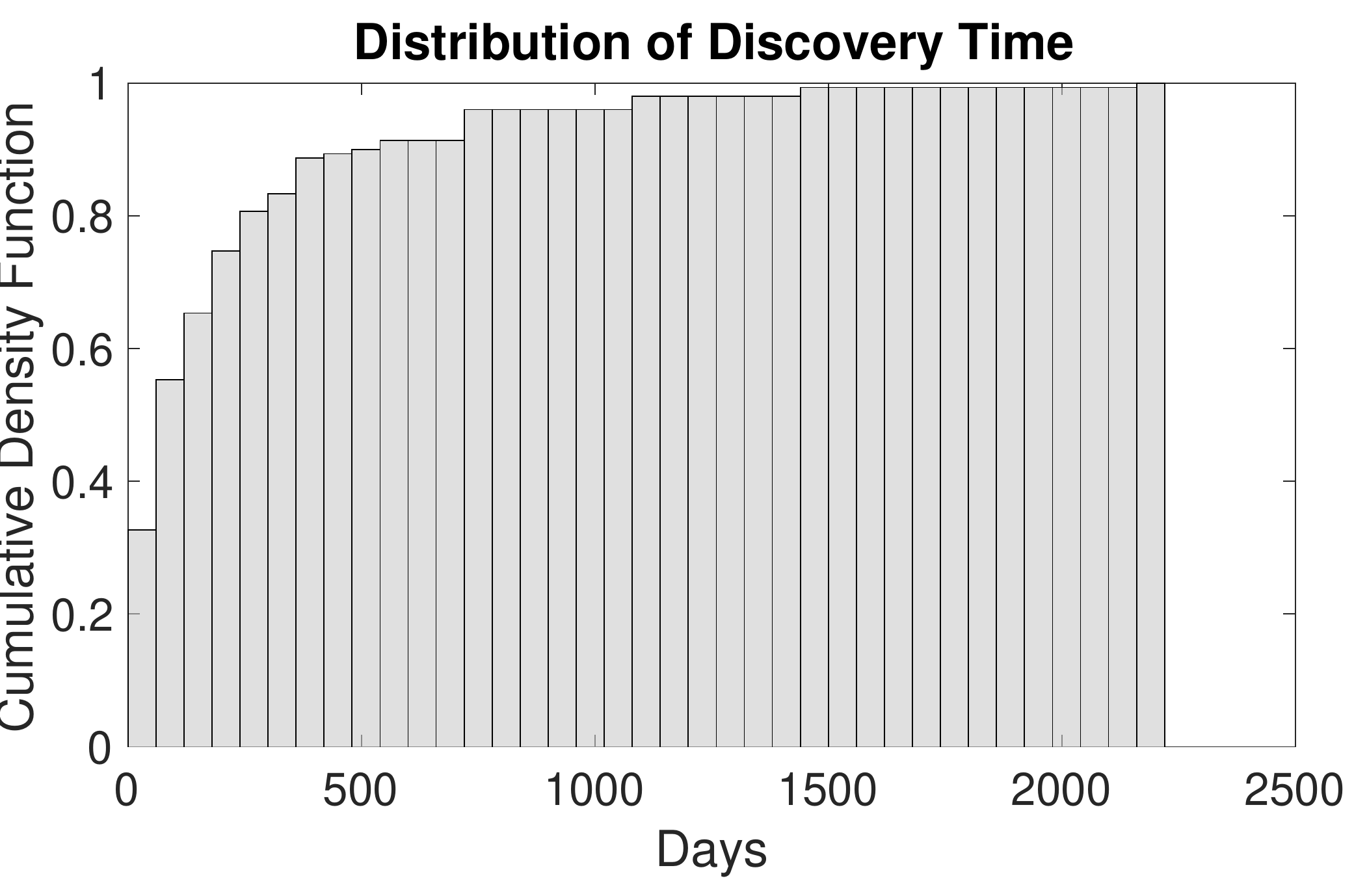}
			\label{fig:a}}
	\end{subfigure}
	\begin{subfigure} [Discovery Time]{ %
			\includegraphics[scale=0.35]{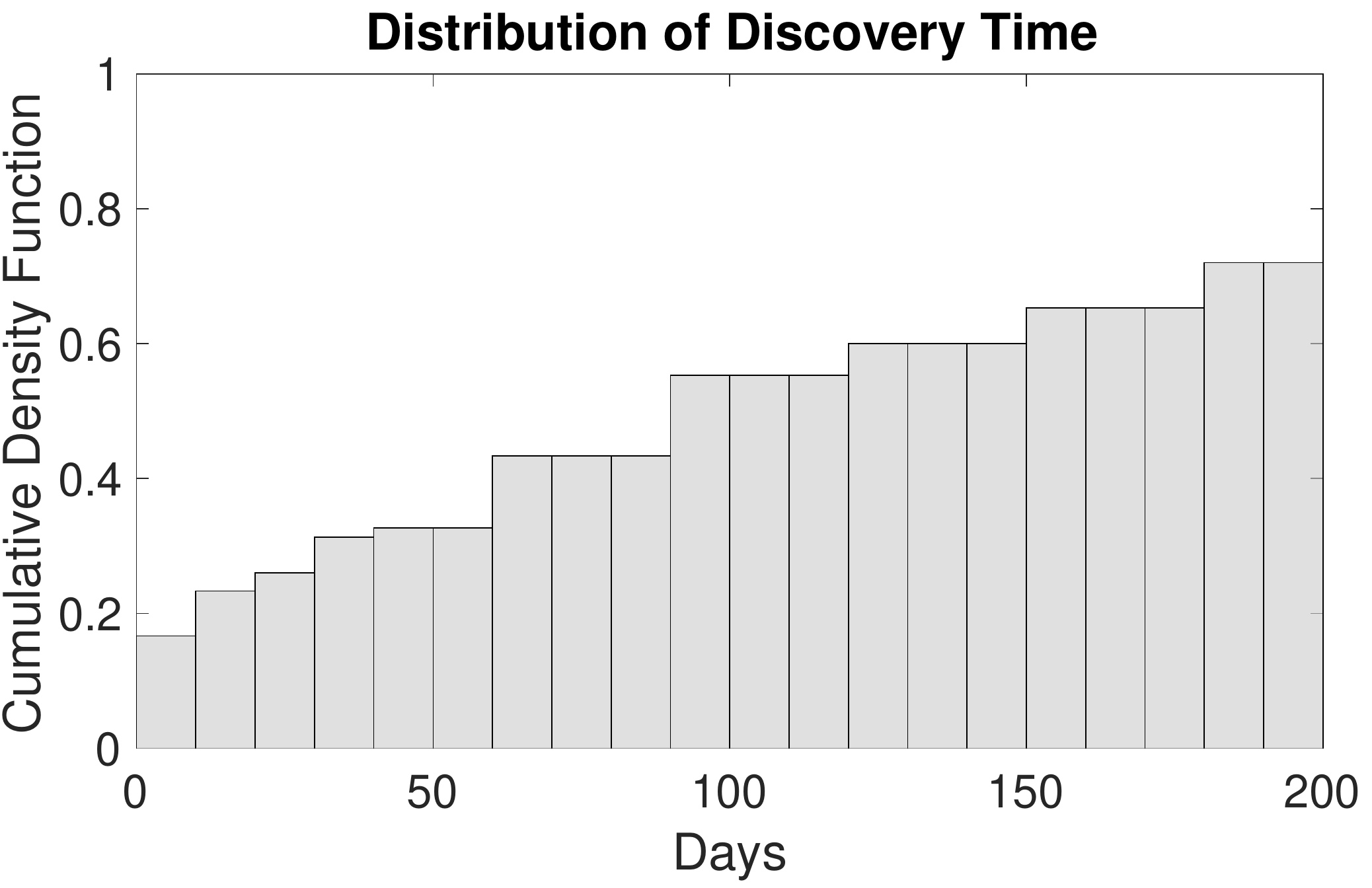}
			\label{fig:b}}
	\end{subfigure}
	\caption{Distribution of discovery time} 
	\label{fig:DiscoveryTime}
\end{figure}

In Figure~\ref{fig:a}, the length of each bar is 60 days. The vertical axis represents the cumulative density of discovery time. As an example, Figure~\ref{fig:a} shows that the discovery time for $32.67\%$ of all attacks is within $60$ days. 
%\sadegh{What do you think about the italic sentence Jens?}\textit{One potential interpretation can be as follows. If the defender spends 60 days in detection or discovery process, the defender can detect the attack with probability $0.3267$ given that the defender's system is compromised.} The more the defender spends in detection, the higher the probability that the defender detects a compromise. 
In Figure~\ref{fig:b}, we take a closer look at the distribution of the discovery time below the average time of the discovery, i.e., $198.2539$ days. At this threshold, about $72\%$ of all attacks in the dataset are detected. In this figure, the length of each bar is equal to $10$ days. For example, we observe that within $20$ days, defenders detect about $23.33\%$ of all attacks. Despite the limitations, the distribution of the discovery time gives us an opportunity to calculate the probability of attack discovery over time.  

For 150 entries, we also have data stating discovery times. Among those, 17 provide also figures for the impact of the attack (see Table~\ref{tab:impact}). As we see in Table~\ref{tab:impact}, losses associated with legal and regulatory causes have the most impact compared to other categories of loss.\footnote{Note that one entry (i.e., with incident time \textit{10/2014}) has a discovery time of 6 years. According to the VERIS definition, incident time should reflect the point in time when compromise first occurred. But, it seems that for this entry, incident time does not reflect this value accurately. Our interpretation is that for this event the incident time likely represents the date at which the incident was discovered.}

%j = [10, 32, 44, 47, 58, 65, 82, 90, 97, 108, 113 (from CN), 120, 133] (incidents with impact fields besides overall rating)

\begin{table}
	\centering
	\begin{adjustbox}{width=1\textwidth}
		\small
\begin{tabular}{|c|c|c|c|c|c|c|c|c|}\hline
Incident Time & Discovery Time & Employee & Asset and Fraud  & Business Disruption & Operating Costs & Legal and Regulatory & Response and Recovery & Overall Amount \\ \hline 

5/2005 & 18 Months & Over 100,000 & 68,000,000 & - & - & 9,700,000   & 256,000,000 & - \\ \hline

2007   & 2 Months & Large & 137,000 & - & - & -  & - & 100,000 \\ \hline

12/2007 & 10 Months & 1001 to 10000 & -  & - & - & 140,000,000  &  - & - \\ \hline

7/2011 & 10 Days & 1001 to 10000 & -  & - & - &  508,000 & - & - \\ \hline
 
9/2011  & 2 Years & 11 to 100 & - & - & - & 3,000,000  & - & - \\ \hline

2/2012  & 5 Months & 25001 to 50000 & - & - & - & 2,725,000  & - & - \\ \hline

2012  & 1 Year & 10001 to 25000 & - & - & - &  - & - & 72,000,000 \\ \hline

4/2013  & 2 Weeks & ``Small'' & - & - & - &  325,000 & - & 325,000 \\ \hline

7/2013 & 15 Days & 10001 to 25000 & - &  2,100,000 & - & -  &  1,600,000 & 3,700,000 (min) \\ \hline

7/2013 & 10 Months & Unknown & - &-  & - & -  &  - & 1,000,000 \\ \hline

11/2013 & 1 Months & Over 100,000 &  - & 148,000,000 & - & 77,000,000  & 18,000,000 & - \\ \hline  

 6/2014  & 1 Month & Over 100000 & - & - & 250,000,000 & - &  - & - \\ \hline    

10/2014   & 6 Years & 1001 to 10000 & - & - & - &  - &  177,000 &  \\ \hline 

4/2015   & 1 Years & 1001 to 10000 & -  & - & - & -  & - &  133,300,000\\ \hline 

7/2015   & 3 Weeks & Unknown & - & - & - & -  & - &  170,000 \\ \hline 
   
8/2015   & 19 Months & 10001 to 25000 & - & - & - & 17,300,000  & - & - \\ \hline
   
8/2015   & 2 Years & 1001 to 10000 & - & - & - & 2,600,000  & - & 2,600,000 \\ \hline

\end{tabular}
\end{adjustbox}
\caption{Impact of security incidents for entries with discovery time. In VCDB some of the entries provide the currency unit, while others do not. All of the above incidents are in US $\$$, except for the incident with incident time $7/2015$ which occurred in China and does not specify the currency unit.}
\label{tab:impact}
\end{table}

\subsection{Protection Time and Reaction Time}
\label{sub:Prot&Reac}

When considering ``Hacking'' and/or ``Malware'' entries for exfiltration time (and excluding data marked with empty, ``NA'', and ``Unknown''), we are left with 44 entries; however, only 5 provide concrete values (see Table~\ref{tab:exfilTime}).

Exfiltration time can be construed as the \textit{protection time} for data compromise events. While the VCDB does not provide much data about exfiltration times, the data allows for a very preliminary first approximation, i.e., that exfiltration (protection) time is lower than discovery time. 

For containment time, we find 258 entries with non-empty fields. After excluding ``Unknown'' and ``NA'' entries, we are left with 175 incidents of which 59 have specific values. One of these entries has a containment time value of 10 years. This incident took place in the year 2000 with a discovery time of 4 years. The summary of the breach in the VCDB is ``Nortel was the victim of a years-long network security breach that allowed hackers to extract its trade secrets, according to a veteran of the bankrupt Canadian telco systems biz.'' We exclude this entry from our data for calculating the distribution of the containment time.  

\begin{table}
	\centering
	\begin{tabular}{|c|c|c|c|} \hline
		Incident Time & Discovery Time & Exfiltration Time & Containment Time \\ \hline
		4/16/2011 & Days & 2 Days & Days \\ \hline
		7/18/2011 & 10 Days & 7 Days & - \\ \hline
		7/24/2013 & 15 Days & 2 Days & - \\ \hline
		11/15/2013 & 1 Months & 2 Weeks & - \\ \hline		
		4/15/2015  & 1 Year & 2 Months & 15 Days \\ \hline
	\end{tabular}
	\caption{Exact value of exfiltration time}
	\label{tab:exfilTime}
\end{table}

Figure~\ref{fig:ContTime} shows the distribution of these 58 entries. In Figure~\ref{fig:c} the length of each bar is 2 days. As we see in this figure, the containment time is no more than 90 days in these 58 entries. The average value of containment time is equal to $10.4504$ days. In Figure~\ref{fig:d}, we take a closer look at the containment times which are lower than the average value. Here we can observe that the containment time is within 1 day for $22.41\%$ of these incidents. Note that a short containment time also implies a short discovery time (since containment time indicates initial compromise to containment/restoration).

\begin{figure}
	\centering
	\begin{subfigure} [Containment Time]{ %
	\includegraphics[scale=0.35]{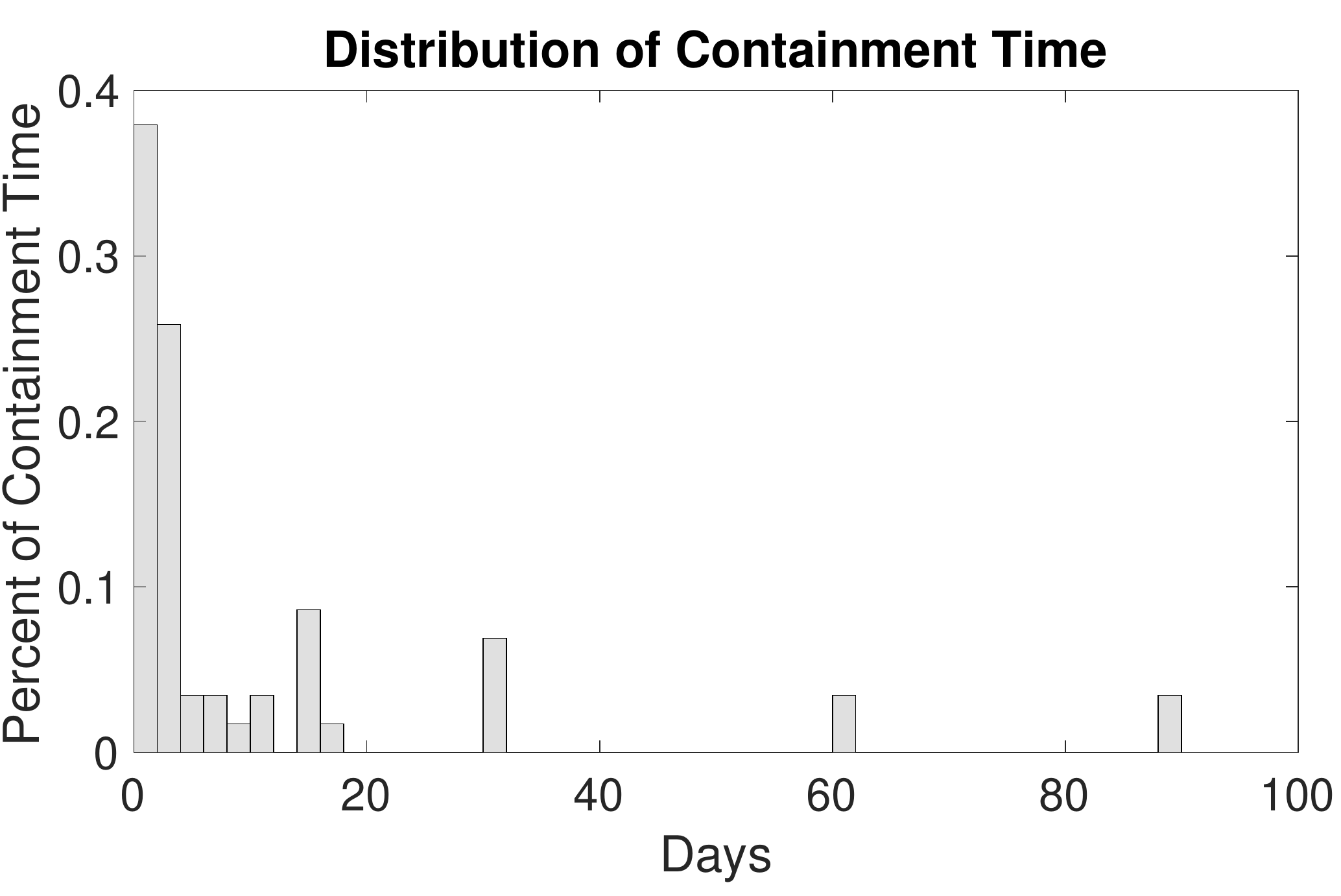}
	\label{fig:c}}
    \end{subfigure}
    \begin{subfigure} [Containment Time]{ %
   			\includegraphics[scale=0.4]{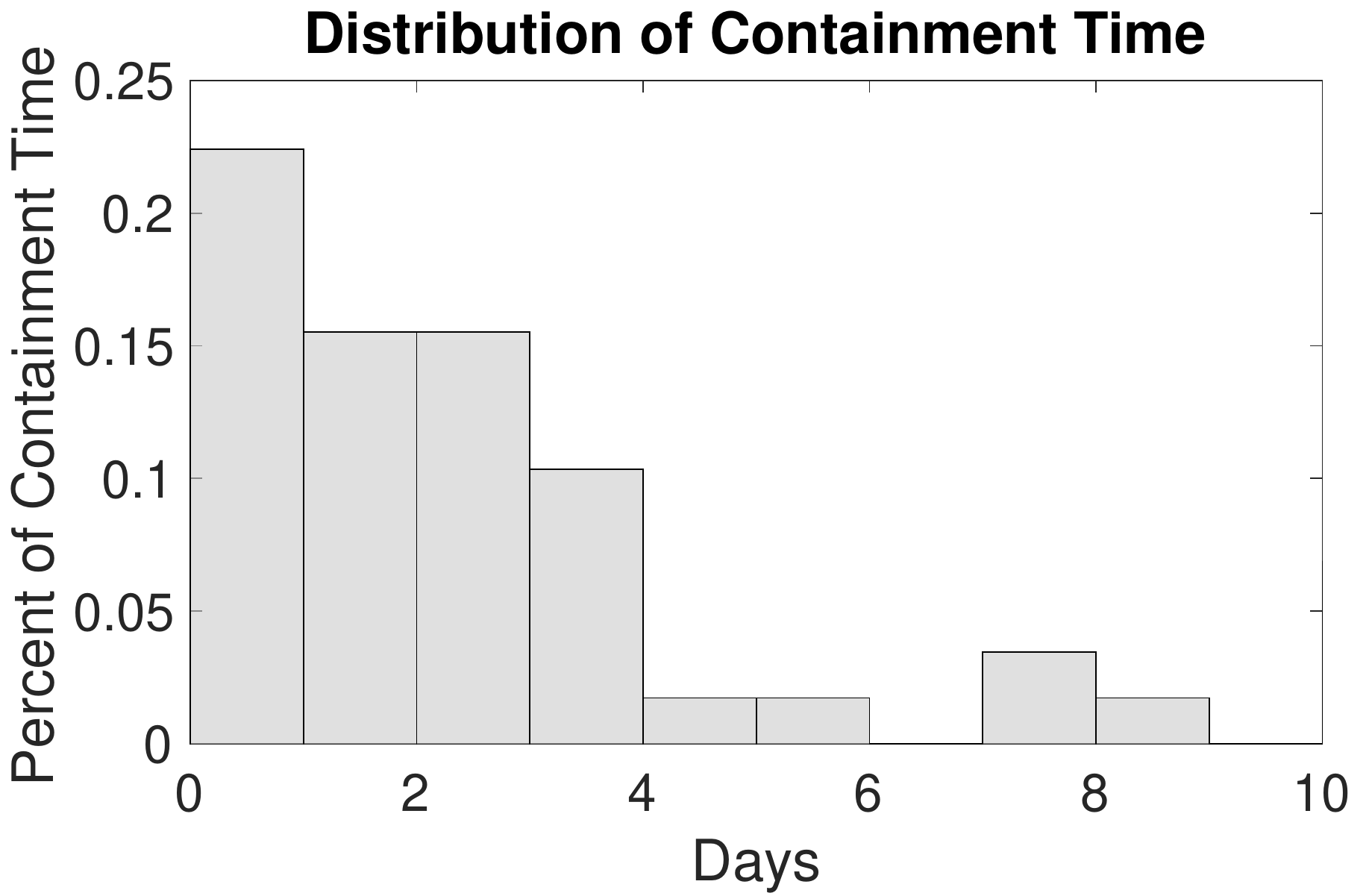} 			
   			\label{fig:d}}
    	\end{subfigure}
	\caption{Distribution of Containment Time}
	\label{fig:ContTime}
\end{figure}

\begin{table}
	\centering
	\begin{tabular}{|c|c|c|c|}\hline
		Incident Time & Discovery Time & Containment Time & Exfiltration Time \\ \hline
		5/4/2001 & Minutes & 15 Days & - \\ \hline 
		
		2004 & 2 Months & 1 Month & - \\ \hline
		
		5/2005 & 18 Months & 2 Months & - \\ \hline 
		
		4/17/2011 & Days & 2 Days & - \\ \hline
		
		9/2012 & 6 Months & 3 Months & - \\ \hline
		
		2/15/2013 & 4 Days & 1 Day & - \\ \hline 
		4/13/2013 & 2 Weeks & 2 Days & - \\ \hline
		6/24/2013 & 3 Days & 3 Days & - \\ \hline
		7/2013 & 6 Months & 8 Days & - \\ \hline
		8/26/2013 & Minutes & 1 Hour & - \\ \hline
		8/29/2013 & Hours & 1 Day & - \\ \hline
		9/3/2013 & 2 Days & 1 Day & - \\ \hline
		9/8/2013 & Minutes & 8 Hours & - \\ \hline
		10/5/2013 & Minutes & 9 Hours & - \\ \hline
		11/2013 & Seconds & 2 Months & - \\ \hline
		12/8/2013 & Seconds & 3 Hours & - \\ \hline
		12/28/2013 & 6 Months & 2 Weeks & - \\ \hline
		
		2/18/2014 & Seconds & 2 Days & - \\ \hline
		3/23/2014 & Seconds & 6 Hours & - \\ \hline
		3/24/2014 & Seconds & 2 Hours & - \\ \hline 
		6/16/2014 & 6 Weeks & 3 Months  & - \\ \hline
		3/27/2014 & Seconds & 90 Minutes & - \\ \hline
		7/3/2014 & Minutes & 1 Day & - \\ \hline
		8/1/2014 & Minutes & 2 Hours & - \\ \hline
		9/1/2014 & 2 Years & 1 Month & - \\ \hline
		10/6/2014 & 6 Years & 2 Days & - \\ \hline

		2/14/2015 & 3 Months & 3 Days & - \\ \hline
		4/15/2015 & 1 Year & 15 Days & 2 Months  \\ \hline
		5/22/2015 & 1 Month & 1 Day & - \\ \hline
		5/2015 & Minutes & 2 Weeks & - \\ \hline
		6/15/2015 & Minutes & 3 Days & - \\ \hline
		7/31/2015 & 3 Weeks & 2 Days & - \\ \hline
		8/2015 & 2 Years & 1 month & - \\ \hline
		11/24/2015 & 10 Days & 16 Days & - \\ \hline
	\end{tabular}
	\caption{Containment Time}
	\label{tab:ContTime}
\end{table}

Table~\ref{tab:ContTime} shows the values of containment time alongside with incident time, discovery time, and exfiltration time. According to the definition of the containment time in VERIS, we expected that the value of the containment time should be higher than the discovery time. (Containment time is equal to initial compromise to containment/restoration, while discovery time is the duration of initial compromise to incident discovery.) But, as we can observe in Table~\ref{tab:ContTime}, in many incidents, the value of containment time is lower than the value of discovery time. We believe that many reporters interpreted containment time as the time that a victim organization spent to recover its system from an incident. Therefore, the value of the containment time may actually more adequately reflect the notion of \textit{reaction time} in our model. 

It is worth mentioning that in the context of phishing, it is possible to calculate the average of the reaction time according to Moore and Clayton~\cite{moore2007examining}, who calculate the mean and median values of phishing sites' lifetime in hours. 

By studying the VCDB and only focusing on ``Hacking'' and ``Malware'' incidents, we are limited to a small set of values for  discovery time. Unfortunately, for the other two important factors, i.e., protection time and the reaction time, the VCDB does not provide information in an obvious fashion. Figure~\ref{fig:ContTime} may suggest that the reaction time is very fast, but there are two issues. First, the number of entries with containment time is quite small. Second, containment time is the sum of discovery time and reaction time and we do not have the discovery time for these entries. 

In order to seek additional insights into the values of protection time, discovery time, and reaction time, we also studied the Web Hacking Incidents Database (WHID)~\cite{WebHacking} and the dataset published by the Privacy Rights Clearinghouse~\cite{PrivacyRights}. None of these two databases provide information about these values. Further, Kuypers et al.~\cite{kuypers2016empirical} investigate the statistical characteristics of sixty thousand cybersecurity incidents for one large organization over the course of six years. In this dataset, each incident has a field showing the duration of man-hours that were required to investigate that incident. The authors also suggest that this dataset includes the time spent for remediation. We anticipate, if this dataset would be made publicly available, one would potentially calculate the value of reaction time as well as discovery time for this specific organization. Another report by Damballa~\cite{Damballa} demonstrated that the typical gap between malware release and detection/remediation using antivirus is 54 days. The study was comprised of over 200,000 malware samples scanned by a leading industry antivirus tool over six months. The study also revealed that almost half of the 200,000 malware samples were not detected on the day they were received, and 15\% of the samples remained undetected after 180 days.

In summary, while on the first glance the VCDB provides a significant amount of data for cybersecurity incidents, the actual details with respect to timing information are insufficient to draw robust conclusions. More special domain data sources, e.g., for phishing \cite{moore2007examining}, may exist, but we are unaware of any larger efforts to collect timing data. We consider this state-of-affairs a significant omission of cybersecurity-related data collection and want to encourage further work in this direction.

At the same time, the lack of empirical data emphasizes the importance of theoretical models to understand the timing aspects of strategic security scenarios. In what follows, we describe our model of Time-Based Security to advance this research field.
%%%%%%%%%%%%%%%%%%%%%%%%%%%%%%%%%%%%%%%  
\section{System Model}
\label{sec:SysMod}
In this section, we propose our model for the Time-Based Security (TBS) approach following the tradition of game theory.
Our model is an infinite two-player game between a \textit{Defender} ($\mathbf{D}$) and an \textit{Attacker} ($\mathbf{A}$) competing with each other to control the defender's resource for a large portion of time, while incorporating the key characteristics derived from TBS.
Note that each player's action to change the ownership of the resource is costly. The attacker's cost to compromise the defender's system is represented by $c_{\mathbf{A}}$. Likewise, $c_{\mathbf{D}}$ denotes the defender's cost to reset the state of the system from compromised to safe. Further, we differentiate between the defender's move to check the state of the resource and the defender's move to reset the resource to a safe state. In what follows, we call the former the defender's \textit{check} and the latter one the defender's \textit{reset}.  We denote the defender's cost to discover whether its system has been compromised, i.e., check, as $c_\mathbf{k}$. 

In this paper, we focus our analysis on the case of a periodically acting attacker and a periodically acting defender. The defender's periodic resource checking, although stationary in nature, partially addresses the defender's uncertainty regarding resource management by creating predictable schedules. For example, it is easier for system administrators to deal with periodic comprehensive security risk evaluations. Further, Farhang and Grossklags~\cite{farhang2016flip} provide two data examples, Microsoft's security policy updates and Oracle's critical patch updates, to show that in practice, several major software vendor organizations update their security policies in a periodic manner.
%Predictable schedules not only make the determination of necessary resources easy, but also balance the operational costs associated with patching vulnerabilities and security risks arising from not patching vulnerabilities. Periodic updates are even more possible today because most vendors release their patches periodically. Although the common practice in the past was to release patches as soon as they were ready, in an effort to ease the burden on system administrators struggling with frequency of updates, and to make the process more predictable, today many software vendors follow periodic patch-release policies.
%We make these simplifying assumptions for our first investigation of TBS given its heightened complexity compared with other games of timing models. 
We denote the periodicity of the attacker attempting to compromise the system, and the defender to check whether the system has been compromised with $t_{\mathbf{A}}$ and $t_{\mathbf{D}}$, respectively. 
%\sadegh{Do we need to add figures of GameSec2016 for explanation of periodic move?}.

%To further simplify the model, we only consider the cost for the defender's reset in our model, and defer the analysis of adding the cost of a check to future work. We account, therefore, for scenarios in which the cost of resetting dwarfs the cost of checking the state of the resource. %The reason is that the cost of the defender's check is lower than the cost of the defender's reset and consequently, we assume that it is equal to zero. 

While the defender checks the resource in a periodic manner, the defender's reset is conditioned on an attack's detection. The defender requires $\mathbf{d}$ amount of time to detect that the system is compromised after the attacker spends $\mathbf{p}$ amount of time to execute the attack successfully. Note that here we propose a pessimistic model for the defender. The defender cannot detect the attacks that are in progress. In other words, only if the defender starts the discovery process after the attacker completely compromised the defender's system, the attack will be detected. 
%The average time between the defender's two consecutive resets is represented by $\delta_{\mathbf{D}}$ depending on $t_{\mathbf{D}}$, $t_{\mathbf{A}}$, $\mathbf{p}$, $\mathbf{d}$, and $\mathbf{r}$. %But, the attacker does not have any ability to detect that the defender gets his ownership back. Therefore, she moves according to periodic strategy.
%In practice, the values of $\mathbf{p}$, $\mathbf{q}$, and $\mathbf{r}$ are not constant and these parameters can change and depend on many factors such as the time of launching an attack, time of defense, the utilized vulnerability exploited by the attacker to execute the attack etc. For simplicity, we assume the average values of these parameters. Hence, the defender cannot find the exact time of the attacker's move to launch the attack even if the defender knows the value of $t_{\mathbf{A}}$. Moreover, we assume that the values of $t_{\mathbf{D}}$, $t_{\mathbf{A}}$, $c_{\mathbf{D}}$, $c_{\mathbf{A}}$, $\mathbf{p}$, $\mathbf{q}$, and $\mathbf{r}$ are common knowledge between the two players. 
At the beginning of the game, the defender checks the state of the resource within interval $[0,t_{\mathbf{D}}]$, and the attacker moves within interval $[0,t_{\mathbf{A}}]$ (both with uniform distribution). Hence, each player does not know the exact time of the other player's move; even if they exactly know the values of the parameters forming the game. 

Moreover, we assume that the values of $\mathbf{p}$, $\mathbf{d}$, and $\mathbf{r}$ are constant for the sake of analysis. It is easy to see that from a defender's point of view a system with a small protection time, but large discovery time and large reaction time should be considered unfavorable. Note that in our data analysis in Section~\ref{sec:Data}, we provide a distribution for the \textit{discovery time}. Here, we do not consider the probabilistic nature of this parameter, but we believe our analysis is an important first step to move towards a comprehensive model for the \textit{time-based security} approach. Further, in practice, there is a possibility that after spending a certain amount of time, i.e., $\mathbf{d}$, the defender may not be able to detect an attack. We exclude this possibility from our current model and will consider it in future work. %Further, by assuming constant values for $\mathbf{p}$ and $\mathbf{r}$, we also take a pessimistic perspective from the defender's view point. 

%Before starting both player's payoffs calculations, we consider the role of each parameter on the ownership of the resource. 
%If at the time of the attacker's move the resource is in the safe state, then for $\mathbf{p}$ amount of time after the attacker's move the resource still belongs to the defender due its inherent resilience. For the duration of the detection time, the owner of the resource is the attacker. 
%During the reaction time, $\mathbf{r}$, the resource belongs to the player controlling the resource before the defender's action. 
Note that an attacker's action during the defender's reaction time will not lead to a successful attack. In other words, the attacker's action in this time interval is ineffective. The reasoning behind this consideration is that during the reaction time, the defender's system is changed to a new safe state. As the attacker's action is based on the previous defender's state, it may not be compatible with the updated/reset defender's system. 
%Note that during the reaction time, the owner of the resource is the player who controls the resource before the defender's last move.

%Similarly, we assume periodic moves for the defender with $t_{\mathbf{D}}$ time between his two consecutive moves. 
According to our model description and parameters, for the rational attacker we have $t_{\mathbf{A}} \geq \mathbf{p+d+r}$, since the defender's moves are conditioned on detection ($\mathbf{d}$ amount of time after the attack's success) and the attacker's attack will be successful $\mathbf{p}$ units of time after the attack. Further, if the defender moves right after the detection, the resource still belongs to the attacker for $\mathbf{r}$ units of time. Therefore, $t_{\mathbf{A}}$ should be higher than or equal to $\mathbf{p+d+r}$. Similarly, for the defender, we also have $t_{\mathbf{D}} \geq \mathbf{p+d+r}$. 
%Further, we restrict our model to the scenario where $t_{\mathbf{D}} \geq t_{\mathbf{A}}$; as the defender is under constant attack, it is not practical to check the state of the resource faster than the attacker's moves. 

To calculate both players' average payoff functions, we need to derive the average time that each player controls the resource minus the average cost of each player's actions over time. In doing so, the general payoff functions are as follows:

\begin{equation}
u_{\mathbf{D}} \left( t_{\mathbf{D}} , t_{\mathbf{A}} \right) = \tau_{\mathbf{D}i} - \frac{c_{\mathbf{D}}}{\delta_{\mathbf{D}i}} - \frac{c_\mathbf{k}}{t_{\mathbf{D}}},
\label{eq:GeneralPayDef}
\end{equation}

\begin{equation}
u_{\mathbf{A}} \left( t_{\mathbf{D}} , t_{\mathbf{A}} \right) = \left(1 - \tau_{\mathbf{D}i} \right)- \frac{c_{\mathbf{A}}}{t_{\mathbf{A}}}. 
\label{eq:GeneralPayAtt}
\end{equation}

Where $\tau_{\mathbf{D}i}$ represents the average fraction of time that the defender controls the resource. It is obvious that the attacker controls the resource for the rest of the time, i.e., $1-\tau_{\mathbf{D}i}$. We use subscript $i$ to differentiate among different cases in our payoff calculations. The time between two consecutive resets by the defender is denoted by $\delta_{\mathbf{D}}$ and the defender's average cost rate over time is equal to $\left(c_{\mathbf{D}}/\delta_{\mathbf{D}i} + c_\mathbf{k}/t_{\mathbf{D}}\right)$. Likewise, the average cost rate for the attacker is equal to $c_{\mathbf{A}}/t_{\mathbf{A}}$. 

We identify \textbf{six} different cases which we discuss in the following. To calculate the payoff functions for these cases, we need to calculate $\tau_{\mathbf{D}}$ and $\delta_{\mathbf{D}}$. Table~\ref{table:hl_value} summarizes the notations, we have used in our model.

 \begin{table}[t]
 	\centering
 	\caption{Summary of notations}\label{tab:Notation}
 	{\rowcolors{2}{white}{lightgray!40}
 		\begin{tabular}{ll}
 			\hline
 			Variable & Definition \\
 			\hline
 			$\mathbf{p}$ & Protection time \\
 			$\mathbf{d}$ & Detection/discovery time \\
 			$\mathbf{r}$ & Reaction time \\
 			$c_{\mathbf{D}}$ & Defender's cost to reset the system's state \\
 			$c_\mathbf{k}$  & Defender's cost to check the state of the system \\
 			$c_{\mathbf{A}}$ & Attacker's cost to compromise the defender \\
 			$t_\mathbf{D}$ & Time between two consecutive checks by the defender \\
 			$t_{\mathbf{A}}$ & Time between two consecutive moves by the attacker \\
 			$\tau_{\mathbf{D}}$ & Average fraction of time that the defender controls the resource \\
 			$\delta_{\mathbf{D}}$ & Average time between the defender's two consecutive reset moves \\
 			$u_{\mathbf{D}}$ & Defender's utility \\
 			$u_{\mathbf{A}}$ & Attacker's utility \\
 			\hline
 		\end{tabular}
 	}
 	\label{table:hl_value}
 \end{table}
 
\textbf{Case 1}: $ t_{\mathbf{D}}  \leq  t_{\mathbf{A}} -\mathbf{p-d-r} $ 
%($ t_{\mathbf{D}} + \mathbf{p+d+r} \leq  t_{\mathbf{A}}  $ )

The above condition also implies that $t_{\mathbf{D}} < t_{\mathbf{A}}$ and the defender's discovery move occurs at least once between the attacker's two consecutive moves. Consider a given attacker move interval $[t,t+t_{\mathbf{A}}]$. In this interval, the defender's discovery move can occur either during the protection time or after the protection time which results in two sub-cases.

\textbf{Case 1.1}: Let $x = \frac{\mathbf{p}}{t_{\mathbf{D}}}$. The probability that the defender's discovery move occurs during the protection time is equal to $x$. According to our game definition, the defender cannot detect this attack. However, the defender's next discovery move occurs before the attacker's next attack, since $ t_{\mathbf{D}} + \mathbf{p+d+r} \leq  t_{\mathbf{A}}  $. Further, both the defender's detection and effectiveness of the reset occur before the next move by the attacker. Therefore, in each attacker move interval, the defender only resets the resource once, i.e., $\delta_{\mathbf{D}11} = t_\mathbf{A}$.

Further, in order to calculate our payoff functions, we need to calculate the average fraction of time that each player controls the resource. 
Consider a given attacker move interval $[t,t+t_{\mathbf{A}}]$. The defender's discovery move occurs uniformly with probability $x$ during the protection time. This discovery move of the defender does not result in the detection and the attacker is the owner of the resource until the successful discovery and the corresponding defensive reset action. Therefore, the average time that the attacker controls the resource is equal to $T_{\mathbf{A}11} = t_{\mathbf{D}} + \mathbf{d+r} - \frac{\mathbf{p}}{2}$.\footnote{The defender's first discovery move occurs in interval $[t,t+\mathbf{p}]$ uniformly. The next one occurs in $[t+t_{\mathbf{D}},t+t_{\mathbf{D}}+\mathbf{p}]$ and leads to the discovery of the attack. The attacker is the owner of the resource after the protection time until the defender's reset move becomes effective, i.e., the defender's defensive move is effective in the interval $[t+t_{\mathbf{D}}+\mathbf{d+r},t+t_{\mathbf{D}}+\mathbf{p+d+r}]$. Due to the defender's uniform move in the protection time, the attacker is the owner of the resource for $t_{\mathbf{D}} + \mathbf{d+r} + \frac{\mathbf{p}}{2} - \mathbf{p}$ on average.} For the remainder of the time, the resource belongs to the defender, i.e., $T_{\mathbf{D}11} = t_\mathbf{A} - T_{\mathbf{A}11}$. Dividing these two values by $t_{\mathbf{A}}$ gives the \textit{average fraction of time} that each player controls the resource. Therefore, we have $\tau_{\mathbf{D}11} = \frac{T_{\mathbf{D}11}}{t_\mathbf{A}}$.

\textbf{Case 1.2}: The probability that the defender's discovery move occurs after the protection time is equal to $1-x$. Thus, the defender can detect and recover its system completely before the launch of the next attack. Hence, in each attacker's interval of two consecutive moves, the defender only resets the resource once, i.e., $\delta_{\mathbf{D}12} = t_\mathbf{A}$.

The attacker is the owner of the resource after the protection time until the discovery move results in detection and the defensive move results in complete system recovery (that is, reaction time has passed). The average time that the attacker controls the resource is $T_{\mathbf{A}12} = \frac{t_{\mathbf{D}}- \mathbf{p}}{2} + \mathbf{d+r}$.\footnote{The defender's discovery move occurs in the interval $[t+\mathbf{p},t+t_{\mathbf{D}}]$ uniformly. Given that the defender's discovery move is distributed uniformly at random, half of this time interval belongs to the attacker plus the detection and the reaction time.} 
The rest of the time, the resource belongs to the defender, i.e., $T_{\mathbf{D}12} = t_\mathbf{A} - T_{\mathbf{A}12}$. Dividing these two values by $t_{\mathbf{A}}$ gives the \textit{average fraction of time} that each player controls the resource. Therefore, we have $\tau_{\mathbf{D}12} = \frac{T_{\mathbf{D}12}}{t_\mathbf{A}}$.

For this case, we have considered two sub-cases. To combine these two sub-cases, we take the expected value with respect to the probability of each sub-case. Therefore, we have:

\begin{equation}
\delta_{\mathbf{D}1} = x \delta_{\mathbf{D}11} + \left(1-x\right)\delta_{\mathbf{D}12} = t_\mathbf{A}.
\label{eq:Case1tD}
\end{equation}

The above formula shows that the defender resets the resource only once when the time between two consecutive discovery moves is low enough, i.e., $t_{\mathbf{D}} \leq t_{\mathbf{A}} - \mathbf{p-d-r}$.

In a similar way, to calculate $\tau_{\mathbf{D}1}$ we have:

\begin{multline}
\tau_{\mathbf{D}1} = x \tau_{\mathbf{D}11} + \left(1-x\right)\tau_{\mathbf{D}12} = \frac{\mathbf{p}\left(t_{\mathbf{A}} - t_{\mathbf{D}} + \frac{\mathbf{p}}{2} - \mathbf{d-r} \right)}{t_{\mathbf{D}}t_{\mathbf{A}}} + \\
\frac{\left(t_{\mathbf{D}} - \mathbf{p}\right) \left( t_{\mathbf{A}} - \frac{t_{\mathbf{D}}}{2} + \frac{\mathbf{p}}{2} - \mathbf{d-r} \right)}{t_{\mathbf{D}}t_{\mathbf{A}}} = \frac{t_{\mathbf{A}} - \frac{t_{\mathbf{D}}}{2} - \mathbf{d-r}}{t_{\mathbf{A}}}.
\label{eq:CaseTD}
\end{multline}

It is interesting to observe that in this case, the average fraction of time that the defender possesses the resource, i.e., $\tau_{\mathbf{D}1}$, does not explicitly depend on the \textit{protection} time. However, this case's condition depends on the value of $\mathbf{p}$. The above formula shows that the lower the values of $\mathbf{d}$ and $\mathbf{r}$, the larger the amount of time that the defender controls the resource. In other words, when the defender checks the state of the resource fast enough, i.e., $t_{\mathbf{D}} \leq t_{\mathbf{A}} - \mathbf{p-d-r}$, the defender's utility is not affected by the protection time. It is rather affected by the discovery time and the reaction time. Influencing the latter factors should then be the focus of attention for a rational defender. %As such, for the defender who checks the state of the resource fast, it is more reasonable for that defender to invest in discovery and the reaction time rather than the protection time. 

By incorporating these two equations, i.e., $\delta_{\mathbf{D}1}$ (Equation~\ref{eq:Case1tD}) and $\tau_{\mathbf{D}1}$ (Equation~\ref{eq:CaseTD}), into Equations~\ref{eq:GeneralPayDef} and \ref{eq:GeneralPayAtt}, we can calculate both players' payoff functions. 

\textbf{Case 2}:  $t_{\mathbf{A}} - \mathbf{p-d-r} \leq t_{\mathbf{D}} \leq t_{\mathbf{A}} - \mathbf{d-r}$ 
%($t_{\mathbf{D}} + \mathbf{d+r} \leq t_{\mathbf{A}} \leq t_{\mathbf{D}} + \mathbf{p+d+r}$)

Similar to the previous case, the defender's discovery move occurs either during the protection time or after the protection time. Thus, we consider two sub-cases as follows. 

\textbf{Case 2.1}: Let $x = \frac{\mathbf{p}}{t_{\mathbf{D}}}$. The probability that the defender's discovery move occurs during the protection time is equal to $x$. According to our game definition, the defender cannot detect an attack that is in progress. Now, consider a given attacker move interval $[t,t+t_{\mathbf{A}}]$. The defender's discovery move occurs uniformly with probability $x$ during the protection time interval, i.e., $[t,t+\mathbf{p}]$, which is not effective. The next defender's discovery move occurs in interval $[t+t_{\mathbf{D}},t+t_{\mathbf{D}}+\mathbf{p}]$. This discovery move results in the detection of the attack and initiates the defensive move. The defensive move is effective after the reaction time which is in interval $[t+t_{\mathbf{D}}+\mathbf{d+r},t+t_{\mathbf{D}}+\mathbf{p+d+r}]$. However, some of the defensive moves in this interval are effective before the attacker's next attack and some of them are after the next attacker's move, since $t_{\mathbf{A}} \leq t_{\mathbf{D}} + \mathbf{p+d+r}$. In other words, for the fraction of this interval, the attacker's next move is not effective since it occurs during either the reaction time or the detection time. Therefore, the defender's defensive move occurs once in $t_{\mathbf{A}}$ or $2t_{\mathbf{A}}$. 

The probability that the defender's defensive move occurs only once in each $t_{\mathbf{A}}$ is equal to $a_{21} = \frac{t_{\mathbf{A}} - t_{\mathbf{D}} - \mathbf{d-r}}{\mathbf{p}}$.\footnote{$a_{21}$ represents the fraction of the defender's move during the protection time that results in exactly one defensive move in each $t_{\mathbf{A}}$. The first discovery move is not effective, since it occurs during the protection time. The second discovery move results in detection and the defensive move. Those defensive moves that are in interval $[t+t_{\mathbf{D}}+\mathbf{d+r},t+t_{\mathbf{A}}]$ result in one defensive move in each $t_{\mathbf{A}}$ (the length of this interval is equal to $t_{\mathbf{A}} - t_{\mathbf{D}} - \mathbf{d-r}$).} The rest of the time, i.e., $a_{22} = 1 - a_{21}$, the defender's reset occurs once in each $2t_{\mathbf{A}}$. Therefore, on average, we have $\delta_{\mathbf{D}21} = a_{21}t_{\mathbf{A}} + a_{22}2t_{\mathbf{A}} $. 

To calculate the average fraction of time that each player controls the resource, we differentiate between the cases when the defender's defensive move is effective in $t_{\mathbf{A}}$ and $2t_{\mathbf{A}}$. When this time is equal to $t_{\mathbf{A}}$, consider a given attacker move interval $[t,t+t_{\mathbf{A}}]$.  Then, the attacker is the owner of the resource after the protection time until the defensive move becomes effective. The defensive move is effective in interval $[t+t_{\mathbf{D}}+\mathbf{d+r},t+ t_\mathbf{A}]$ uniformly. On average, the attacker is the owner of the resource for half of this time interval plus the time before the defense effectiveness except the protection time. Hence, we have $T_{\mathbf{A}21_1}= \frac{t_{\mathbf{A}}+t_{\mathbf{D}} + \mathbf{d+r}}{2} - \mathbf{p}$. And the defender is the owner of the resource for the rest of the time, i.e., $T_{\mathbf{D}12_1} = t_{\mathbf{A}} - T_{\mathbf{A}21_1} $. Dividing these two values by $t_{\mathbf{A}}$ gives the average fraction of time that each player controls the resource. Therefore, we have $\tau_{\mathbf{D}21_1} = \frac{T_{\mathbf{D}21_1}}{t_\mathbf{A}}$. When each defensive move occurs in $2t_{\mathbf{A}}$, consider an interval of two consecutive moves for the attacker $[t,t+2t_{\mathbf{A}}]$. In a similar way, we have $T_{\mathbf{A}12_2} = \frac{t_{\mathbf{A}} + t_{\mathbf{D}} + \mathbf{d+r-p}}{2}$,  $T_{\mathbf{D}12_2} = 2t_{\mathbf{A}} - T_{\mathbf{A}21_2} $, and $\tau_{\mathbf{D}21_2} = \frac{T_{\mathbf{D}21_2}}{2t_\mathbf{A}}$. To combine these two cases and calculate $\tau_{\mathbf{D}21}$, we take an average based on the fraction of time that each of these cases occurs, i.e., $\tau_{\mathbf{D}21} = a_{21} \tau_{\mathbf{D}21_1} + a_{22} \tau_{\mathbf{D}21_2}$.

\textbf{Case 2.2}: Here, the defender's discovery move occurs after the protection time. This case is similar to sub-case 1.2, and we therefore omit its detailed representation.

By combining these two sub-cases, we have:

\begin{equation}
\delta_{\mathbf{D}2} =  x \delta_{\mathbf{D}21} + \left(1-x\right)\delta_{\mathbf{D}22} = x \left(a_{21} + 2a_{22}\right) t_{\mathbf{A}} + \left(1-x\right) t_{\mathbf{A}} =
2t_{\mathbf{A}} - \left(\frac{t_{\mathbf{A}} - \mathbf{p-d-r}}{t_{\mathbf{D}}}\right)t_{\mathbf{A}}.
\label{eq:Case2tD}
\end{equation}

According to the above equation, the average time between the defender's two consecutive defensive moves is $t_{\mathbf{A}} \leq \delta_{\mathbf{D}2} \leq 2t_{\mathbf{A}}$ and it depends on $t_{\mathbf{A}}$, $\mathbf{p}$, $\mathbf{d}$, and $\mathbf{r}$. 

\begin{multline}
\tau_{\mathbf{D}2} = x \tau_{\mathbf{D}21} + \left(1-x\right)\tau_{\mathbf{D}22} = x \left(a_{21} \tau_{\mathbf{D}21_1} + a_{22} \tau_{\mathbf{D}21_2}\right) + \left(1-x\right)\tau_{\mathbf{D}22} = \\
\frac{1}{4t_{\mathbf{A}}t_{\mathbf{D}}} \left(-t_{\mathbf{A}}^2-t_{\mathbf{D}}^2 + 4t_{\mathbf{A}}t_{\mathbf{D}} + 2\mathbf{p}t_{\mathbf{A}} -2t_{\mathbf{D}}\left(\mathbf{d+r}\right) + \left(\mathbf{p+d+r}\right) \left(\mathbf{d+r-p}\right) \right).
\label{eq:Case2TD}
\end{multline}

Contrary to Case 1, the defender's average fraction of time for controlling the resource depends on protection time. 

\textbf{Case 3}:  $t_{\mathbf{A}} - \mathbf{d-r} \leq t_{\mathbf{D}} \leq t_{\mathbf{A}}$ %($t_{\mathbf{D}}  \leq t_{\mathbf{A}} \leq t_{\mathbf{D}} + \mathbf{d+r}$)

Similar to the two previous cases, the defender's discovery move occurs either during the protection time or after the protection time. Thus, we consider two sub-cases as follows. 

\textbf{Case 3.1}: Let $x = \frac{\mathbf{p}}{t_{\mathbf{D}}}$. Consider a given attacker move interval $[t,t+t_{\mathbf{A}}]$. The defender's discovery move occurs uniformly with probability $x$ during the protection time interval, i.e., $[t,t+\mathbf{p}]$, which is not effective. The next defender's discovery move occurs in interval $[t+t_{\mathbf{D}},t+t_{\mathbf{D}}+\mathbf{p}]$. This discovery move results in the detection of the attack and initiates the defensive move. The defensive move is effective after the reaction time that is in interval $[t+t_{\mathbf{D}}+\mathbf{d+r},t+t_{\mathbf{D}}+\mathbf{p+d+r}]$. This means that all the defensive moves in this interval are effective after the attacker's next move. Thus, the attacker's next move is not effective. Therefore, in each $2t_{\mathbf{A}}$, the defender resets the resource exactly once, i.e., $\delta_{\mathbf{D}31} = 2t_{\mathbf{A}}$. 

To calculate the average time that the attacker controls the resource, note that the defensive move is effective (after the reaction time) in interval $[t+t_{\mathbf{D}}+\mathbf{d+r},t+t_{\mathbf{D}}+\mathbf{p+d+r}]$ uniformly. Half of this time interval belongs to the attacker on average. Moreover, the attacker is the owner of the resource after the protection time until the defensive move's effectiveness. Therefore, we have $T_{\mathbf{A}31} = t_{\mathbf{D}} + \mathbf{d+r} - \frac{\mathbf{p}}{2}$. The rest of the time, the defender is the owner of the resource, i.e., $T_{\mathbf{D}31} = 2t_{\mathbf{A}} - T_{\mathbf{A}31}$. Thus, we have $\tau_{\mathbf{D}31} = \frac{T_{\mathbf{D}31}}{2t_{\mathbf{A}}}$.

\textbf{Case 3.2}: Consider a given attacker move interval $[t,t+t_{\mathbf{A}}]$. The probability that the defender's discovery move occurs after the protection time interval, i.e., $[t+\mathbf{p}, t+t_{\mathbf{D}}]$, is equal to $1-x$. The defender's discovery move in this interval results in detection and the defensive move. The defensive move is effective in interval $[t+\mathbf{p+d+r}, t+t_{\mathbf{D}}+\mathbf{d+r}]$. A defensive move may be effective after the attacker's next attack, since $t_{\mathbf{A}} \leq t_{\mathbf{D}} + \mathbf{d+r}$, and the attacker's next move will not be successful. Hence, the defender's defensive move occurs once in $t_{\mathbf{A}}$ or $2t_{\mathbf{A}}$. Similar to case 2.1, the fraction of time that the defender's defensive move occurs only once in each $t_{\mathbf{A}}$ is equal to $a_{31} = \frac{t_{\mathbf{A}} - \mathbf{p-d-r}}{t_{\mathbf{D}}-\mathbf{p}}$. The rest of the time, i.e., $a_{32} = 1 - a_{31}$, the defender's defensive move occurs once in each $2t_{\mathbf{A}}$. Therefore, we have $\delta_{\mathbf{D}32} = a_{31}t_{\mathbf{A}} + a_{32}2t_{\mathbf{A}} $.

To calculate the average fraction of time that each player controls the resource, we differentiate between the cases when the defender's defensive move occurs in $t_{\mathbf{A}}$ and $2t_{\mathbf{A}}$. When this time is equal to $t_{\mathbf{A}}$, consider a given attacker move interval $[t,t+t_{\mathbf{A}}]$.  The attacker is the owner of the resource after the protection time until the defensive move becomes effective. The defensive move is effective in interval $[t+\mathbf{p+d+r},t+ t_\mathbf{A}]$ uniformly. On average, the attacker is the owner of the resource for half of this time interval plus the time before the defense effectiveness except the protection time. Hence, we have $T_{\mathbf{A}32_1}= \frac{t_{\mathbf{A}} + \mathbf{p+d+r}}{2} - \mathbf{p}$. And the defender is the owner for the remainder, i.e., $T_{\mathbf{D}32_1} = t_{\mathbf{A}} - T_{\mathbf{A}32_1} $. Dividing these two values by $t_{\mathbf{A}}$ gives the \textit{average fraction of time} that each player controls the resource. Therefore, we have $\tau_{\mathbf{D}32_1} = \frac{T_{\mathbf{D}32_1}}{t_\mathbf{A}}$. 

When the defensive move occurs in each $2t_{\mathbf{A}}$, consider a given two consecutive moves interval for the attacker $[t,t+2t_{\mathbf{A}}]$. In a similar way, we have $T_{\mathbf{A}32_2} = \frac{t_{\mathbf{A}} + t_{\mathbf{D}} + \mathbf{d+r}}{2} - \mathbf{p}$,  $T_{\mathbf{D}32_2} = 2t_{\mathbf{A}} - T_{\mathbf{A}32_2} $, and $\tau_{\mathbf{D}32_2} = \frac{T_{\mathbf{D}32_2}}{2t_\mathbf{A}}$. To combine these two cases and calculate $\tau_{\mathbf{D}32}$, we take an average based on the fraction of time that each of these cases occurs, i.e., $\tau_{\mathbf{D}32} = a_{31} \tau_{\mathbf{D}32_1} + a_{32} \tau_{\mathbf{D}32_2}$.

By combining these two sub-cases we have:

\begin{equation}
\delta_{\mathbf{D}3} =  x \delta_{\mathbf{D}31} + \left(1-x\right)\delta_{\mathbf{D}32} = x2t_{\mathbf{A}}  + \left(1-x\right) \left(a_{31} + 2a_{32}\right) t_{\mathbf{A}} = 
2t_{\mathbf{A}} - \left(\frac{t_{\mathbf{A}} - \mathbf{p-d-r}}{t_{\mathbf{D}}}\right)t_{\mathbf{A}}.
\label{eq:Case3tD}
\end{equation}

One of the boundary points for this case is $t_{\mathbf{A}} = t_{\mathbf{D}}$. Inserting $t_{\mathbf{A}} = t_{\mathbf{D}}$ in the above equation gives that $\delta_{\mathbf{D}3} = t_{\mathbf{D}} + \mathbf{p+d+r}$. One might expect that for $t_{\mathbf{A}} = t_{\mathbf{D}}$, the value of $\tau_{\mathbf{D}3}$ should be equal to $t_{\mathbf{A}}$. Note that here, we want to calculate the average time between the defender's two consecutive defensive moves. When $t_{\mathbf{A}} = t_{\mathbf{D}}$, if the defender's discovery move occurs after the attack success, the defensive move occurs in each $t_{\mathbf{A}}$. But, if the discovery move occurs during the protection time, the defensive move occurs once in each $2t_{\mathbf{A}}$. Therefore, the average value is not equal to $t_{\mathbf{A}}$. 

\begin{multline}
\tau_{\mathbf{D}3} = x \tau_{\mathbf{D}31} + \left(1-x\right)\tau_{\mathbf{D}32} = x \tau_{\mathbf{D}31}  + \left(1-x\right)\left(a_{31} \tau_{\mathbf{D}32_1} + a_{32} \tau_{\mathbf{D}32_2}\right) = \\
\frac{1}{4t_{\mathbf{A}}t_{\mathbf{D}}} \left(-t_{\mathbf{A}}^2-t_{\mathbf{D}}^2 + 4t_{\mathbf{A}}t_{\mathbf{D}} + 2\mathbf{p}t_{\mathbf{A}} -2t_{\mathbf{D}}\left(\mathbf{d+r}\right) + \left(\mathbf{p+d+r}\right) \left(\mathbf{d+r-p}\right) \right).
\label{eq:Case3TD}
\end{multline}

Note that the above two values are the same as what we have calculated for Case 2. It is interesting to see that while the situation for each case's payoff calculation is different, it yields the same values for both cases' payoff functions. 

\textbf{Case 4}: $t_{\mathbf{A}} \leq t_{\mathbf{D}} \leq t_{\mathbf{A}} + \mathbf{p}$

The above condition implies that the attacker moves once or twice in a given discovery move interval, because we have $t_{\mathbf{D}} \geq t_{\mathbf{A}}$ and $t_{\mathbf{D}} \leq t_{\mathbf{A}} + \mathbf{p} < 2t_{\mathbf{A}}$ (note that we assume $t_{\mathbf{A}} \geq \mathbf{p+d+r}$). Consider a given discovery move interval $[t, t+t_{\mathbf{D}}]$. The defender's discovery move at an arbitrary time $t$ results in detection and the defensive move. Based on the attack occurrence during this interval, we consider three sub-cases.

\textbf{Case 4.1}: Let $y_1 = \frac{\mathbf{d+r}}{t_{\mathbf{A}}}$. The probability that the attacker's move occurs within the $\mathbf{d+r}$ amount of time after the discovery move is equal to $y_1$. The attacker's move during this interval does not lead to a successful attack, since the defender's discovery move at time $t$ yields the detection and the defensive move. 

The next attacker's move occurs in interval $[t+t_{\mathbf{A}},t+t_{\mathbf{A}}+ \mathbf{d+r}]$. The attacker has to spend $\mathbf{p}$ amount of time to execute its attack successfully. Thus, all of the attacker's moves in this interval are successful after the defender's next discovery move. The next discovery move does not lead to the detection and the defensive move. Consequently, in each $2t_{\mathbf{D}}$, the defender resets the resource only once, i.e., $\tau_{\mathbf{D}41} = 2t_{\mathbf{D}}$. 

The defender is the owner of the resource after the defensive move is effective, i.e., $t+\mathbf{d+r}$ has elapsed, until the attacker completely compromises the defender's system. The attacker's move effectiveness is uniformly occurring in interval $[t+t_{\mathbf{A}}+\mathbf{p},t+t_{\mathbf{A}}+\mathbf{p+d+r}]$. Due to the uniform distribution, half of this time interval belongs to the defender on average. Therefore, we have $T_{\mathbf{D}41} = t_{\mathbf{A}} + \mathbf{p} - \frac{\mathbf{d+r}}{2}$ and $\tau_{\mathbf{D}41} = \frac{T_{\mathbf{D}41}}{2t_{\mathbf{D}}}$. The rest belongs to the attacker which is equal to $T_{\mathbf{A}41} = 2t_{\mathbf{D}} - T_{\mathbf{D}41}$. 

\textbf{Case 4.2}: Let $y_2 = \frac{t_{\mathbf{D}} - \mathbf{p-d-r}}{t_{\mathbf{A}}}$. The probability that the attacker can successfully compromise the defender's resource completely before the next discovery move is equal to $y_2$.\footnote{If the attacker moves in interval $[t+\mathbf{d+r},t+t_{\mathbf{D}}- \mathbf{p}]$, the attacker's attack is successful before the defender's next discovery move at $t+t_{\mathbf{D}}$. Note that the attacker requires $\mathbf{p}$ amount of time to compromise the defender's system.} In this case, the defender can detect the attack by its next discovery move which means that the defensive move occurs once in each $t_{\mathbf{D}}$, i.e., $\delta_{\mathbf{D}42} = t_{\mathbf{D}}$. 

The defender is the owner of the resource after its defensive move is effective, i.e., $t+\mathbf{d+r}$ has elapsed, until the attacker completely compromises the defender's system. The attacker's move effectiveness is uniformly occurring in interval $[t+\mathbf{p+d+r},t+t_{\mathbf{D}}]$. Due to uniform distribution, half of this interval belongs to the defender. Therefore, we have $T_{\mathbf{D}42} = \frac{\mathbf{t_{\mathbf{D}} + p-d-r}}{2}$ and $\tau_{\mathbf{D}42} = \frac{T_{\mathbf{D}42}}{t_{\mathbf{D}}}$. The rest belongs to the attacker, i.e., $T_{\mathbf{A}42} = t_{\mathbf{D}} - T_{\mathbf{D}42}$.

\textbf{Case 4.3}: Let $y_3 = 1 - y_1 - y_2 = \frac{t_{\mathbf{A}} - t_{\mathbf{D}} +\mathbf{p} }{t_{\mathbf{A}}}$. The probability that the attacker moves in interval $[t+t_{\mathbf{D}}-\mathbf{p}, t+ t_{\mathbf{A}}]$ is equal to $y_3$. The attacker's move in this interval will be successful after the defender's next discovery move. Therefore, we have $\tau_{\mathbf{D}43} = 2t_{\mathbf{D}}$. 

The defender is the owner of the resource after its defensive move's effectiveness, i.e., $t+\mathbf{d+r}$ has passed, until the attacker completely compromises the defender's system. The attacker's move effectiveness is uniformly occurring in interval $[t+t_{\mathbf{D}},t+t_{\mathbf{D}}+\mathbf{p}]$. Due to the uniform distribution, half of this interval belongs to the defender. Therefore, we have $T_{\mathbf{D}43} = \frac{\mathbf{t_{\mathbf{D}} + t_{\mathbf{A}} + p}}{2} - \mathbf{d-r}$ and $\tau_{\mathbf{D}43} = \frac{T_{\mathbf{D}43}}{2t_{\mathbf{D}}}$. The rest belongs to the attacker, i.e., $T_{\mathbf{A}43} = 2t_{\mathbf{D}} - T_{\mathbf{D}43}$.

Similar to the previous cases, for $\delta_{\mathbf{D}4}$ and $\tau_{\mathbf{D}4}$, we have:

\begin{equation}
\delta_{\mathbf{D}4} =  y_1 \delta_{\mathbf{D}41} + y_2\delta_{\mathbf{D}42} + y_3\delta_{\mathbf{D}43} = 
2t_{\mathbf{D}} - \left(\frac{t_{\mathbf{D}} - \mathbf{p-d-r}}{t_{\mathbf{A}}}\right)t_{\mathbf{D}}.
\label{eq:Case4tD}
\end{equation}

\begin{multline}
\tau_{\mathbf{D}4} = y_1 \tau_{\mathbf{D}41} + y_2 \tau_{\mathbf{D}42} + y_3 \tau_{\mathbf{D}43} =  \\
\frac{1}{4t_{\mathbf{A}}t_{\mathbf{D}}} \left(t_{\mathbf{A}}^2+t_{\mathbf{D}}^2 + 2\mathbf{p}t_{\mathbf{A}} -2t_{\mathbf{D}}\left(\mathbf{d+r}\right) + \left(\mathbf{p+d+r}\right) \left(\mathbf{d+r-p}\right) \right).
\label{eq:Case4TD}
\end{multline}

\textbf{Case 5}: $t_{\mathbf{A}} + \mathbf{p} \leq t_{\mathbf{D}} \leq t_{\mathbf{A}} + \mathbf{p+d+r}$

Similar to the previous case, the attacker moves once or twice in a given discovery move interval. For a given discovery move interval $[t,t+t_{\mathbf{D}}]$, the defender's discovery move at time $t$ results in discovery and the defensive move. Based on attack occurrence, we have two sub-cases.

\textbf{Case 5.1}: Let $y_1 = \frac{\mathbf{d+r}}{t_{\mathbf{A}}}$. The probability that the attacker's move occurs within $\mathbf{d+r}$ amount of time after the discovery move is equal to $y_1$. The attacker's move during this interval does not lead to a successful attack, since the defender's discovery move at time $t$ yields to the detection and the defensive move. 

The next attacker's move occurs in interval $[t+t_{\mathbf{A}},t+t_{\mathbf{A}}+ \mathbf{d+r}]$. The attacker has to spend $\mathbf{p}$ amount of time to execute its attack successfully. Thus, part of the attacker's move in this interval are successful after the defender's next discovery move, i.e., interval $[t+t_{\mathbf{D}},t+t_{\mathbf{A}}+ \mathbf{p+d+r}]$, and the next discovery move does not result in detection and the defensive move. Consequently, the defender resets the resource in each $t_{\mathbf{D}}$ or $2t_{\mathbf{D}}$. 
If the attacker's next attack is successful in the interval $[t+t_{\mathbf{A}}+\mathbf{p}, t+t_{\mathbf{D}}]$, the defender's defensive move occurs in each $t_{\mathbf{D}}$. The probability that the attacker's move occurs in this interval in this sub-case, which is distributed uniformly, is equal to $a_{51} = \frac{t_{\mathbf{D}} - t_{\mathbf{A}}-\mathbf{p} }{\mathbf{d+r}}$. On average, half of this time interval belongs to the defender plus the time after the defender's defense effectiveness at time $t + \mathbf{d+r}$ until the successful compromise. Therefore, we have $T_{\mathbf{D}51_1} = \frac{t_{\mathbf{A}}+t_{\mathbf{D}} + \mathbf{p} }{2} - \mathbf{d-r}$ and $\tau_{\mathbf{D}51_1} = \frac{T_{\mathbf{D}51_1}}{t_{\mathbf{D}}}$.

In a similar way, if the attacker's next move is successful in the interval $[t+t_{\mathbf{D}},t+t_{\mathbf{A}}+\mathbf{p+d+r}]$, the defender's defensive move occurs in each $2t_{\mathbf{D}}$. The attacker's move occurs in this interval with probability $a_{52} = 1 - a_{51}$. On average, half of this time interval belongs to the defender plus the time after the defender's defense effectiveness at time $t + \mathbf{d+r}$ until the successful compromise. Therefore, we have $T_{\mathbf{D}51_2} = \frac{t_{\mathbf{A}}+t_{\mathbf{D}} + \mathbf{p-d-r} }{2} $ and $\tau_{\mathbf{D}51_2} = \frac{T_{\mathbf{D}51_2}}{2t_{\mathbf{D}}}$. By combining these two scenarios, we have $\delta_{\mathbf{D}51} = a_{51}t_{\mathbf{D}} + a_{52}2t_{\mathbf{D}}$ and $\tau_{\mathbf{D}51} = a_{51}\tau_{\mathbf{D}51_1} + a_{52}\tau_{\mathbf{D}51_2}$.
 
\textbf{Case 5.2}: Let $y = 1-y_1$. The probability that the attacker moves after the defensive move, i.e., interval $[t+\mathbf{d+r},t+t_{\mathbf{A}}]$, is equal to $y$. The defender can detect the attacks occurring in this interval, since $t_{\mathbf{A}}+\mathbf{p} \leq t_{\mathbf{D}}$. Thus, the defender resets the resource in each $t_{\mathbf{D}}$, i.e., $\delta_{\mathbf{D}52} = t_{\mathbf{D}}$. Similar to the previous cases, we have $T_{\mathbf{D}52} = \frac{t_{\mathbf{A}}-\mathbf{d-r}}{2} + \mathbf{p}$ and $\tau_{\mathbf{D}52} = \frac{T_{\mathbf{D}52} }{t_{\mathbf{D}}}$. 

For $\delta_{\mathbf{D}5}$ and $\tau_{\mathbf{D}5}$, we have:

\begin{equation}
\delta_{\mathbf{D}5} =  y_1 \delta_{\mathbf{D}51} + y\delta_{\mathbf{D}52}  = y_1 \left(a_{51} + 2a_{52}\right)t_{\mathbf{D}} + yt_{\mathbf{D}}
= 2t_{\mathbf{D}} - \left(\frac{t_{\mathbf{D}} - \mathbf{p-d-r}}{t_{\mathbf{A}}}\right)t_{\mathbf{D}},
\label{eq:Case5tD}
\end{equation}

and

\begin{multline}
\tau_{\mathbf{D}5} = y_1 \tau_{\mathbf{D}51} + y \tau_{\mathbf{D}52}  = y_1 \left(a_{51}\tau_{\mathbf{D}51_1} + a_{52}\tau_{\mathbf{D}51_2}\right) + y \tau_{\mathbf{D}52} \\
= \frac{1}{4t_{\mathbf{A}}t_{\mathbf{D}}} \left(t_{\mathbf{A}}^2+t_{\mathbf{D}}^2 + 2\mathbf{p}t_{\mathbf{A}} -2t_{\mathbf{D}}\left(\mathbf{d+r}\right) + \left(\mathbf{p+d+r}\right) \left(\mathbf{d+r-p}\right) \right).
\label{eq:Case5TD}
\end{multline}

Note that both Case 4 and Case 5 result in the same payoff functions. 

\textbf{Case 6}: $t_{\mathbf{D}} \geq t_{\mathbf{A}} + \mathbf{p+d+r}$

Consider a given discovery move interval $[t,t+t_{\mathbf{D}}]$. The attacker moves in the interval $[t+\mathbf{d+r},t+t_{\mathbf{D}}-\mathbf{p}]$ at least once, since $t_{\mathbf{D}} \geq t_{\mathbf{A}} + \mathbf{p+d+r}$. Hence, the attacker's move in this interval is detected in each $t_{\mathbf{D}}$ and we have $\delta_{\mathbf{D}6} = t_{\mathbf{D}}$. 

After the defensive move occurs successfully, i.e., $t+\mathbf{d+r}$ has passed, the defender is the owner of the resource until the successful compromise occurs. The attacker moves in interval $[t+\mathbf{d+r},t+t_{\mathbf{A}}+\mathbf{d+r}]$ with uniform distribution. Half of this time interval belongs to the defender plus the protection time. In other words, we have $T_{\mathbf{D}6} = \frac{t_{\mathbf{A}}}{2} + \mathbf{p}$ and $\tau_{\mathbf{D}6} = \frac{T_{\mathbf{D}6}}{t_{\mathbf{D}}}$.

\begin{equation}
\delta_{\mathbf{D}6} =  t_{\mathbf{D}}.
\label{eq:Case6tD}
\end{equation}

\begin{equation}
\tau_{\mathbf{D}6} = \frac{t_{\mathbf{A}}+2\mathbf{p}}{2t_{\mathbf{D}}}.
\label{eq:Case6TD}
\end{equation}

According to the equation above, the average fraction of time does not depend on discovery time and the reaction time when the time between two consecutive discovery moves is high enough, i.e., $t_{\mathbf{D}} \geq t_{\mathbf{A}} + \mathbf{p+d+r}$. This means that if the defender does not check its system state regularly, in order to increase its utility, the defender should only invest in increasing the protection time. Other parameters are not important. Furthermore, according to Equation~\ref{eq:Case6tD}, the average time between the defender's two consecutive resets is equal to the time between two consecutive discovery moves. 

%%%%%%%%%%%%%%%%%%%%%%%%%%%%%%%%%%%%%%%    
%%\section{Payoff Model}
%%\label{sec:Payoff}
%%\input{Payoff}
%%%%%%%%%%%%%%%%%%%%%%%%%%%%%%%%%%%%%%%    
\section{Analytical Results}
\label{sec:Ana}
In this section, we analyze our proposed game and, in particular, provide both players' best responses.  

\newtheorem{thm}{Theorem}

In order to calculate the defender's best response, first, we define the following three points in Definition~\ref{Def:DefenderPlausible1}. 

\newtheorem{Def}{Definition}
\begin{Def} For given $t_{\mathbf{A}}$, points $\bar{t}_{\mathbf{D}1}$, $\bar{t}_{\mathbf{D}2}$, and $\bar{t}_{\mathbf{D}3}$ are calculated as follows:

- $\bar{t}_{\mathbf{D}1} = \sqrt{2t_{\mathbf{A}} c_\mathbf{k}}$ if $\bar{t}_{\mathbf{D}1} \geq \mathbf{p+d+r} $ and $\bar{t}_{\mathbf{D}1} \leq t_{\mathbf{A}}- \mathbf{p-d-r}$.

- The solution of the following equation in $t_{\mathbf{D}}$, i.e., $\bar{t}_{\mathbf{D2}}$, if $\bar{t}_{\mathbf{D}2} \geq \mathbf{p+d+r}$ and $t_{\mathbf{A}} - \mathbf{p-d-r} \leq \bar{t}_{\mathbf{D2}} \leq t_{\mathbf{A}}$:

\begin{multline}
\frac{c_\mathbf{k}}{t_{\mathbf{D}}^2} + \frac{c_{\mathbf{D}} \left(t_{\mathbf{A}} - \mathbf{p-d-r}\right)  }{t_{\mathbf{A}} \left( 2t_{\mathbf{D}} - t_{\mathbf{A}} + \mathbf{p+d+r} \right)^2 } \\
+\frac{1}{4t_{\mathbf{A}}t_{\mathbf{D}}^2} \left( -t_{\mathbf{D}}^2 + t_{\mathbf{A}}^2 -2\mathbf{p}t_{\mathbf{A}} -  \left(\mathbf{p+d+r}\right) \left(\mathbf{d+r-p}\right)  \right) = 0.
\label{eq:Case3Sol}
\end{multline}

- The solution of the following equation in $t_{\mathbf{D}}$, i.e., $\bar{t}_{\mathbf{D3}}$, if $t_{\mathbf{A}} \leq \bar{t}_{\mathbf{D3}} \leq t_{\mathbf{A}} + \mathbf{p+d+r}$:

\begin{multline}
\frac{c_\mathbf{k}}{t_{\mathbf{D}}^2} + \frac{c_{\mathbf{D}} t_{\mathbf{A}}  }{t_{\mathbf{D}}^2 \left( 2t_{\mathbf{A}} - t_{\mathbf{D}} + \mathbf{p+d+r} \right)} -  \frac{c_{\mathbf{D}} t_{\mathbf{A}}  }{t_{\mathbf{D}} \left( 2t_{\mathbf{A}} - t_{\mathbf{D}} + \mathbf{p+d+r} \right)^2 }\\
+\frac{1}{4t_{\mathbf{A}}t_{\mathbf{D}}^2} \left( t_{\mathbf{D}}^2 - t_{\mathbf{A}}^2 -2\mathbf{p}t_{\mathbf{A}} -  \left(\mathbf{p+d+r}\right) \left(\mathbf{d+r-p}\right)  \right) = 0
\label{eq:Case4Sol}
\end{multline}

\label{Def:DefenderPlausible1}
\end{Def}

In the above definition, each point represents the critical point of a case. $\bar{t}_{\mathbf{D}1}$ represents the defender's critical point for case 1. $\bar{t}_{\mathbf{D}2}$ represents the defender's critical point for cases 2 and 3. $\bar{t}_{\mathbf{D}3}$ is the critical point for the defender's payoff in cases 4 and 5. These points in addition to the boundary points of each case provide the entire set of possible best responses by the defender. 
%give a set in which the defender's best response is one of them.

\begin{Def}
The members of set $\mathcal{S}(t_{\mathbf{A}})$ are defined as follows:

$\mathcal{S}(t_{\mathbf{A}}) = \{ \bar{t}_{\mathbf{D1}}, \bar{t}_{\mathbf{D2}} , \bar{t}_{\mathbf{D3}} , \mathbf{p+d+r}, t_{\mathbf{A}} - \mathbf{p-d-r}, t_{\mathbf{A}}  , t_{\mathbf{A}} + \mathbf{p+d+r}\}$. 
\label{Def:DefenderPlausible}
\end{Def}

Theorem~\ref{thm:BRdefender} represents the defender's best response.

\begin{thm}
For each value of $t_{\mathbf{A}}$, the defender's best response is calculated as follows:
\begin{equation}
BR_{\mathbf{D}} (t_{\mathbf{A}}) = \arg\max_{t_{\mathbf{D}} \in \mathcal{S}} u_{\mathbf{D}} \left( t_{\mathbf{D}} , t_{\mathbf{A}} \right).
\label{eq:DefBR}
\end{equation}
\label{thm:BRdefender}
\end{thm}

\begin{proof}
To show our results, we identify the defender's maximum points for each case and then compare the resulting payoffs of these points to each other to find the point yielding the highest payoff. This point is the defender's best response for a corresponding $t_{\mathbf{A}}$. For each case, we take the partial derivative from Equation~\ref{eq:GeneralPayDef} with respect to $t_{\mathbf{D}}$ which is represented in the following equation and set the partial derivative to zero to find critical points. 

\begin{equation}
\frac{\partial u_{\mathbf{D}} (t_{\mathbf{D}},t_{\mathbf{A}})}{\partial t_{\mathbf{D}}} = \frac{\partial \tau_{\mathbf{D}i}}{\partial t_{\mathbf{D}}} + \frac{c_{\mathbf{D}}}{\left( \delta_{\mathbf{D}i} \right)^2} \left( \frac{\partial \delta_{\mathbf{D}i}}{\partial t_{\mathbf{D}i}} \right) + \frac{c_\mathbf{k}}{t_{\mathbf{D}}^2}.
\label{eq:partialDer3}
\end{equation}

\textbf{Case 1}: Setting Equation~\ref{eq:partialDer3} to zero gives $\bar{t}_{\mathbf{D}1}$. The defender's payoff function in case 1 is increasing in $[0, \bar{t}_{\mathbf{D}1}]$ and decreasing in $[\bar{t}_{\mathbf{D}1}, \infty]$. Hence, the defender's payoff function is maximized at $median\{\mathbf{p+d+r},\bar{t}_{\mathbf{D}1}, t_{\mathbf{A}}-\mathbf{p-d-r}\}$, where median is a middle value of a set.

%%%%%%%%%%%%%%%%%%%%%%%%%%%%

\textbf{Case 2 and Case 3}:  
%By calculating $t'_{\mathbf{D}3}$, its partial derivative with respect to $t_{\mathbf{D}}$, partial derivative of $T_{\mathbf{D}3}$ with respect to $t_{\mathbf{D}}$, and 
Setting Equation~\ref{eq:partialDer3} to zero gives Equation~\ref{eq:Case3Sol}. The solution of this equation, i.e., $\bar{t}_{\mathbf{D}2}$, gives the critical point(s) of the defender's payoff function if $t_{\mathbf{A}} -\mathbf{p-d-r} \leq \bar{t}_{\mathbf{D2}} \leq t_{\mathbf{A}}$. Therefore, we should compare the defender's payoff at $\bar{t}_{\mathbf{D}2}$ with $t_{\mathbf{A}}$ and $t_{\mathbf{A}} - \mathbf{p-d-r}$ to find the local maximum.

%%%%%%%%%%%%%%%%%%%%%%%%%%%%

\textbf{Case 4 and Case 5}: In a similar way, by taking the partial derivative from the defender's payoff function in Case 4 and Case 5 with respect to $t_{\mathbf{D}}$ and setting it to zero we have Equation~\ref{eq:Case4Sol}. The solution of this equation, i.e., $\bar{t}_{\mathbf{D}3}$ is extremum if $t_{\mathbf{A}} \leq \bar{t}_{\mathbf{D3}} \leq t_{\mathbf{A}} + \mathbf{p+d+r}$. In order to find the maximum for this case, we compare the defender's payoffs at $t_{\mathbf{A}}$, $t_{\mathbf{A}}+\mathbf{p+d+r}$, and $\bar{t}_{\mathbf{D}3}$ with each other. The point with the highest payoff is the maximum for this case.

%%%%%%%%%%%%%%%%%%%%%%%%%%%

\textbf{Case 6}: In this case, the defender's payoff is decreasing in $t_{\mathbf{D}}$. Thus, the defender's payoff is maximized at $t_{\mathbf{A} }+ \mathbf{p+d+r}$. 

%%%%%%%%%%%%%%%%%%%%%%%%%%%%%

By comparing the resulting payoffs of different cases, the one with the highest payoff provides the defender's best response for given value of $t_{\mathbf{A}}$.
% $\qed$
\end{proof}

In order to calculate the attacker's best response, we follow an equivalent approach as utilized for the defender's best response (see Theorem~\ref{thm:BRdefender}). First, we identify the attacker's possible best responses for each case, and then compare them in order to identify the attacker's best response.
% (yielding a non-negative payoff). 

\begin{Def}
	For given $t_{\mathbf{D}}$, points $\bar{t}_{\mathbf{A}1}$, $\bar{t}_{\mathbf{A}2}$, and $\bar{t}_{\mathbf{A}3}$ are calculated as follows:
	
	- $\bar{t}_{\mathbf{A}1} = \sqrt{ 2t_{\mathbf{D}} c_{\mathbf{A}} }$ if $\bar{t}_{\mathbf{A}1} \geq \mathbf{p+d+r}$. 
	
	- The solution of the following equation in $t_{\mathbf{A}}$, i.e., $\bar{t}_{\mathbf{A}2}$, if $t_{\mathbf{D}} \leq \bar{t}_{\mathbf{A}1} \leq t_{\mathbf{D}} + \mathbf{p+d+r}$:
	
	\begin{equation}
	\frac{c_{\mathbf{A}}}{t_{\mathbf{A}}^2} - \frac{ t_{\mathbf{D}}^2 - t_{\mathbf{A}}^2 + 2\left(\mathbf{d+r}\right)t_{\mathbf{D}} - \left(\mathbf{p+d+r}\right)\left(\mathbf{d+r-p}\right) }{4t_{\mathbf{A}}^2 t_{\mathbf{D}}} =0.
	\label{eq:Attcase1}
	\end{equation}
	
	- The solution of the following equation in $t_{\mathbf{A}}$, i.e., $\bar{t}_{\mathbf{A}3}$, if $\bar{t}_{\mathbf{A}3} \geq \mathbf{p+d+r}$ and $t_{\mathbf{D}}-\mathbf{p-d-r} \leq \bar{t}_{\mathbf{A}3} \leq t_{\mathbf{D}}$:
	
	\begin{equation}
		\frac{c_{\mathbf{A}}}{t_{\mathbf{A}}^2} - \frac{ t_{\mathbf{A}}^2 - t_{\mathbf{D}}^2 + 2\left(\mathbf{d+r}\right)t_{\mathbf{D}} - \left(\mathbf{p+d+r}\right)\left(\mathbf{d+r-p}\right) }{4t_{\mathbf{A}}^2 t_{\mathbf{D}}} =0.
	\label{eq:Attcase4}
	\end{equation}
	
\end{Def}

In the above definition, each point represents the critical point of a case. $\bar{t}_{\mathbf{A}1}$ represents the attacker's critical point for case 6. $\bar{t}_{\mathbf{A}2}$ represents the attacker's critical point for cases 2 and 3. $\bar{t}_{\mathbf{A}3}$ is the critical point for the attacker's payoff in cases 4 and 5. These points in addition to the boundary points of each case provide the entire set of possible best responses by the defender.

\begin{Def}
	The members of set $\mathcal{V}(t_{\mathbf{D}})$ are defined as follows:
	
	$\mathcal{V}(t_{\mathbf{D}}) = \{ \bar{t}_{\mathbf{A1}}, \bar{t}_{\mathbf{A2}} , \bar{t}_{\mathbf{A3}} , \mathbf{p+d+r}, t_{\mathbf{D}} - \mathbf{p-d-r}, t_{\mathbf{D}}  , t_{\mathbf{D}} + \mathbf{p+d+r}\}$. 
	\label{Def:AttPlausible}
\end{Def}

Theorem~\ref{thm:BRAtt} represents the attacker's best response.

\begin{thm}
	For each value of $t_{\mathbf{D}}$, the attacker's best response is calculated as follows:
	\begin{equation}
	BR_{\mathbf{A}} (t_{\mathbf{D}}) = \arg\max_{t_{\mathbf{A}} \in \mathcal{V}} u_{\mathbf{A}} \left( t_{\mathbf{D}} , t_{\mathbf{A}} \right).
	\label{eq:AttBR}
	\end{equation}
	\label{thm:BRAtt}
\end{thm}

\begin{proof}
	To show our results, we identify the attacker's maximum points for each case and then compare the resulting payoffs of these points to each other to find the point yielding the highest payoff. This point is the attacker's best response for a corresponding $t_{\mathbf{D}}$. For each case, we take the partial derivative from Equation~\ref{eq:GeneralPayAtt} with respect to $t_{\mathbf{A}}$ which is represented in the following equation and set the partial derivative to zero to find critical points. 
	
	\begin{equation}
	\frac{\partial u_{\mathbf{A}} (t_{\mathbf{D}},t_{\mathbf{A}})}{\partial t_{\mathbf{A}}} = - \frac{\partial \tau_{\mathbf{D}i}}{\partial t_{\mathbf{D}}} + \frac{c_{\mathbf{A}}}{t_{\mathbf{A}}^2}.
	\label{eq:partialDerAtt}
	\end{equation}

	\textbf{Case 1}: In this case, the attacker's payoff is decreasing in $t_{\mathbf{A}}$. Thus, the attacker's payoff is maximized at $t_{\mathbf{D} }+ \mathbf{p+d+r}$.
	
	%%%%%%%%%%%%%%%%%%%%%%%%%%%%
	
	\textbf{Case 2 and Case 3}:  
	%By calculating $t'_{\mathbf{D}3}$, its partial derivative with respect to $t_{\mathbf{D}}$, partial derivative of $T_{\mathbf{D}3}$ with respect to $t_{\mathbf{D}}$, and 
	Setting Equation~\ref{eq:partialDerAtt} to zero gives Equation~\ref{eq:Attcase1}. The solution of this equation, i.e., $\bar{t}_{\mathbf{A}2}$, gives the critical point(s) of the attacker's payoff function if $t_{\mathbf{D}} \leq \bar{t}_{\mathbf{A}2} \leq t_{\mathbf{D}} + \mathbf{p+d+r}$. Therefore, we should compare the attacker's payoff at $\bar{t}_{\mathbf{A}2}$ with $t_{\mathbf{D}}$ and $t_{\mathbf{D}} + \mathbf{p+d+r}$ to find the local maximum.
	
	%%%%%%%%%%%%%%%%%%%%%%%%%%%%
	
	\textbf{Case 4 and Case 5}: In a similar way, by taking the partial derivative from the attacker's payoff function in Case 4 and Case 5 with respect to $t_{\mathbf{A}}$ and setting it to zero gives Equation~\ref{eq:Attcase4}. The solution of this equation, i.e., $\bar{t}_{\mathbf{A}3}$ is extremum if $t_{\mathbf{D}} - \mathbf{p-d-r} \leq \bar{t}_{\mathbf{A}3} \leq t_{\mathbf{D}}$. In order to find the maximum for this case, we compare the defender's payoffs at $t_{\mathbf{D}}$, $t_{\mathbf{D}}-\mathbf{p-d-r}$, and $\bar{t}_{\mathbf{A}3}$ with each other. The point with the highest payoff is the maximum for this case.
	
	%%%%%%%%%%%%%%%%%%%%%%%%%%%
	
	\textbf{Case 6}:  Setting Equation~\ref{eq:partialDerAtt} to zero gives $\bar{t}_{\mathbf{A}1}$. The attacker's payoff function in Case 1 is increasing in $[0, \bar{t}_{\mathbf{A}1}]$ and decreasing in $[\bar{t}_{\mathbf{A}1}, \infty]$. Hence, the attacker's payoff function is maximized at $median\{\mathbf{p+d+r},\bar{t}_{\mathbf{A}1}, t_{\mathbf{D}}-\mathbf{p-d-r}\}$, where median is a middle value of a set.
	
	%%%%%%%%%%%%%%%%%%%%%%%%%%%%%
	
By comparing the resulting payoff of all cases, the one with the highest payoff provides the attacker's best response for given value of $t_{\mathbf{A}}$.
% $\qed$	
\end{proof}

%%%%%%%%%%%%%%%%%%%%%%%%%%%%%%%%%%%%%%%
\section{Numerical Illustration}
\label{sec:Num}
In this section, we evaluate our findings numerically. 
First, we study the effect of $\mathbf{p}$, $\mathbf{d}$, and $\mathbf{r}$ on both players' best responses. Second, we consider the effect of $c_{\mathbf{D}}$ and $c_\mathbf{k}$ on the defender's best response. Third, we investigate the role of $c_{\mathbf{A}}$ on the attacker's best response. Then, we investigate the existence of Nash equilibria in our proposed game. 

\begin{figure}
	\centering
	\begin{subfigure} [Defender's best response]{ %
			\includegraphics[scale=0.35]{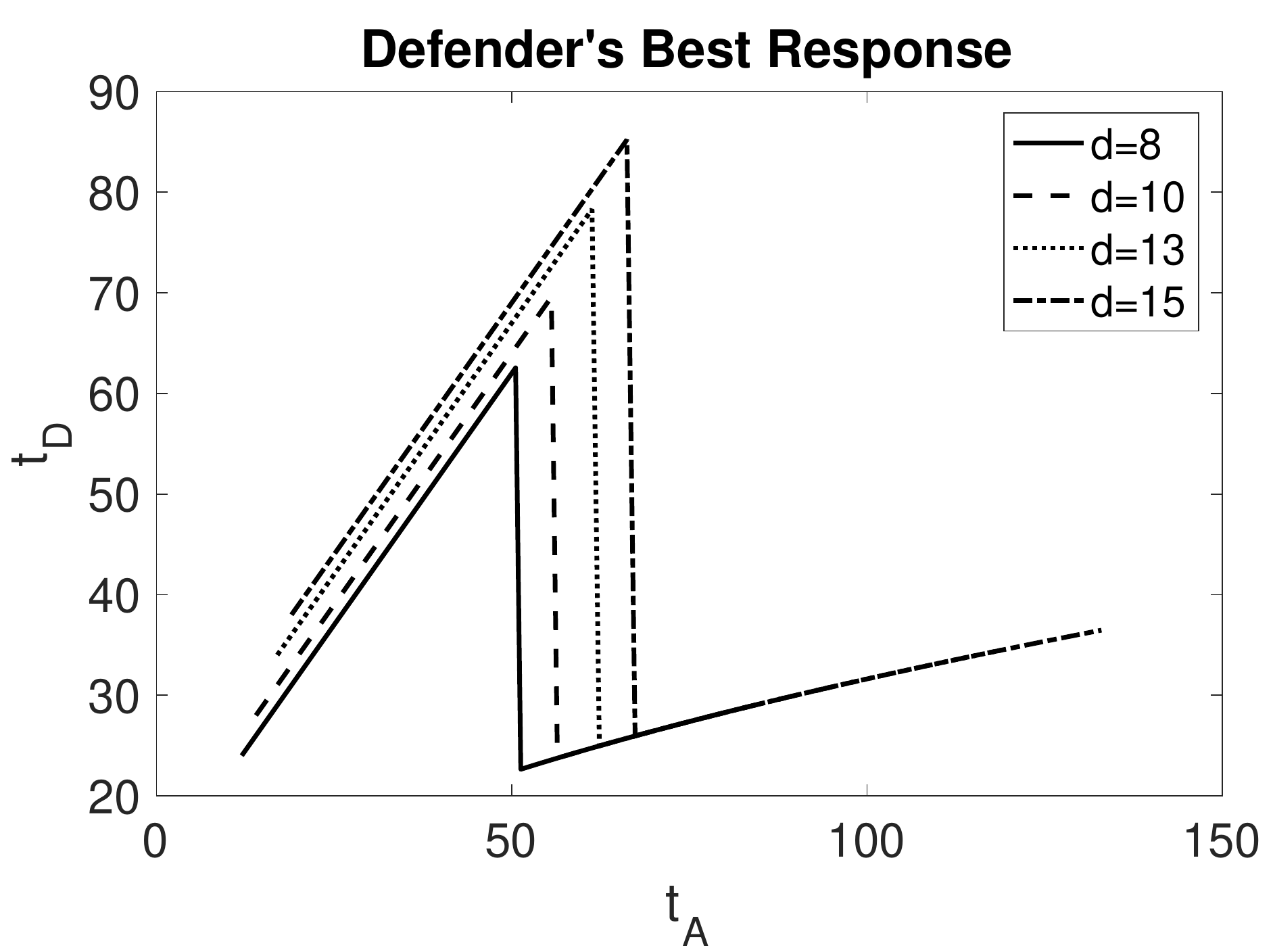}
			\label{fig:BrDp}}
	\end{subfigure}
	\begin{subfigure} [Attacker's best response]{ %
			\includegraphics[scale=0.35]{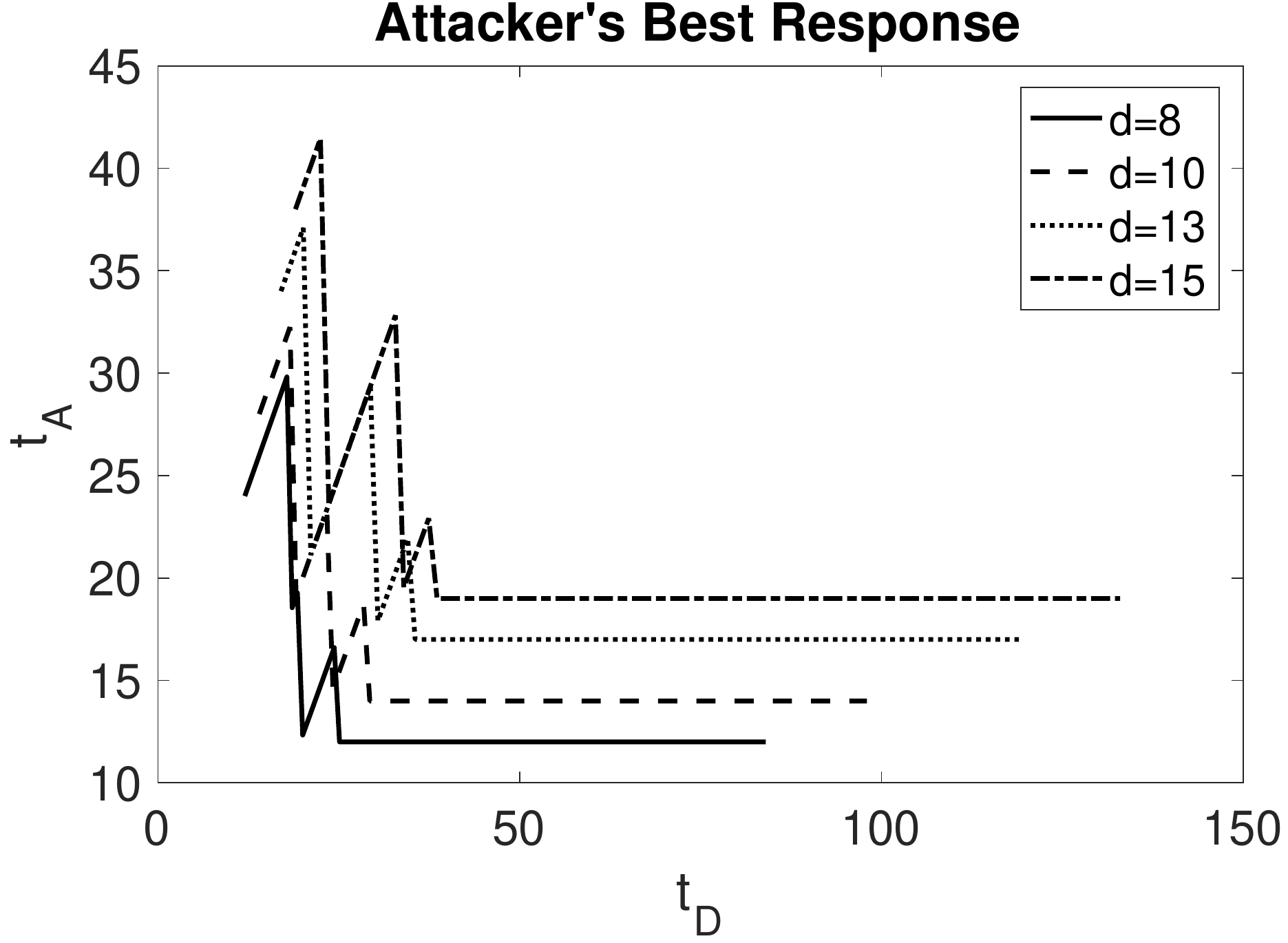} 			
			\label{fig:BrAp}}
	\end{subfigure}
	\caption{Players' best responses for different values of $\mathbf{d}$. We have $c_{\mathbf{D}}=2$, $c_\mathbf{k}=5$, $c_{\mathbf{A}} = 0.5$, $\mathbf{p}=3$, and $\mathbf{r}=1$. }
	\label{fig:BRp}
\end{figure}

Figures~\ref{fig:BrDp} and \ref{fig:BrAp} represent the defender's and the attacker's best responses for different values of $\mathbf{d}$, respectively. For these two plots, we have $\mathbf{p}=3$, $\mathbf{r}=1$, $c_{\mathbf{D}}=2$, $c_\mathbf{k}=5$, and $c_{\mathbf{A}} = 0.5$. 
Note that we assume that $t_{\mathbf{A}} \geq \mathbf{p+d+r}$ and $t_{\mathbf{D}} \geq \mathbf{p+d+r}$. For each curve, we plot each player's best response in interval $[\mathbf{p+d+r},7\left(\mathbf{p+d+r}\right)]$. Therefore, each curve starts and ends at different points.
Based on Figure~\ref{fig:BrDp}, for low values of $t_{\mathbf{A}}$, the defender's best response is equal to $t_{\mathbf{A}}+\mathbf{p+d+r}$. The higher the detection time, the less often the defender checks the state of its resource. But, when $t_{\mathbf{A}}$ is large, i.e., the attacker moves slowly, the defender's best response is equal to $\sqrt{2c_\mathbf{k} t_{\mathbf{A}}}$, which is independent from $\mathbf{p}$, $\mathbf{d}$, and $\mathbf{r}$. For the attacker's best response, when $t_{\mathbf{D}}$ is small, the attacker's best response is equal to $t_{\mathbf{D}}+\mathbf{p+d+r}$. For higher values of $t_{\mathbf{D}}$, the attacker's best response is equal to $t_{\mathbf{D}}$ and for even higher values, the attacker's best response is the solution of Equation~\ref{eq:Attcase4}. If $t_{\mathbf{D}}$ is high enough, the attacker's best response is equal to $\mathbf{p+d+r}$. Therefore, the attacker moves slower, i.e., higher values of $t_{\mathbf{A}}$, when $\mathbf{d}$ is higher. 

\begin{figure}
	\centering
	\begin{subfigure} [Defender's best response]{ %
			\includegraphics[scale=0.35]{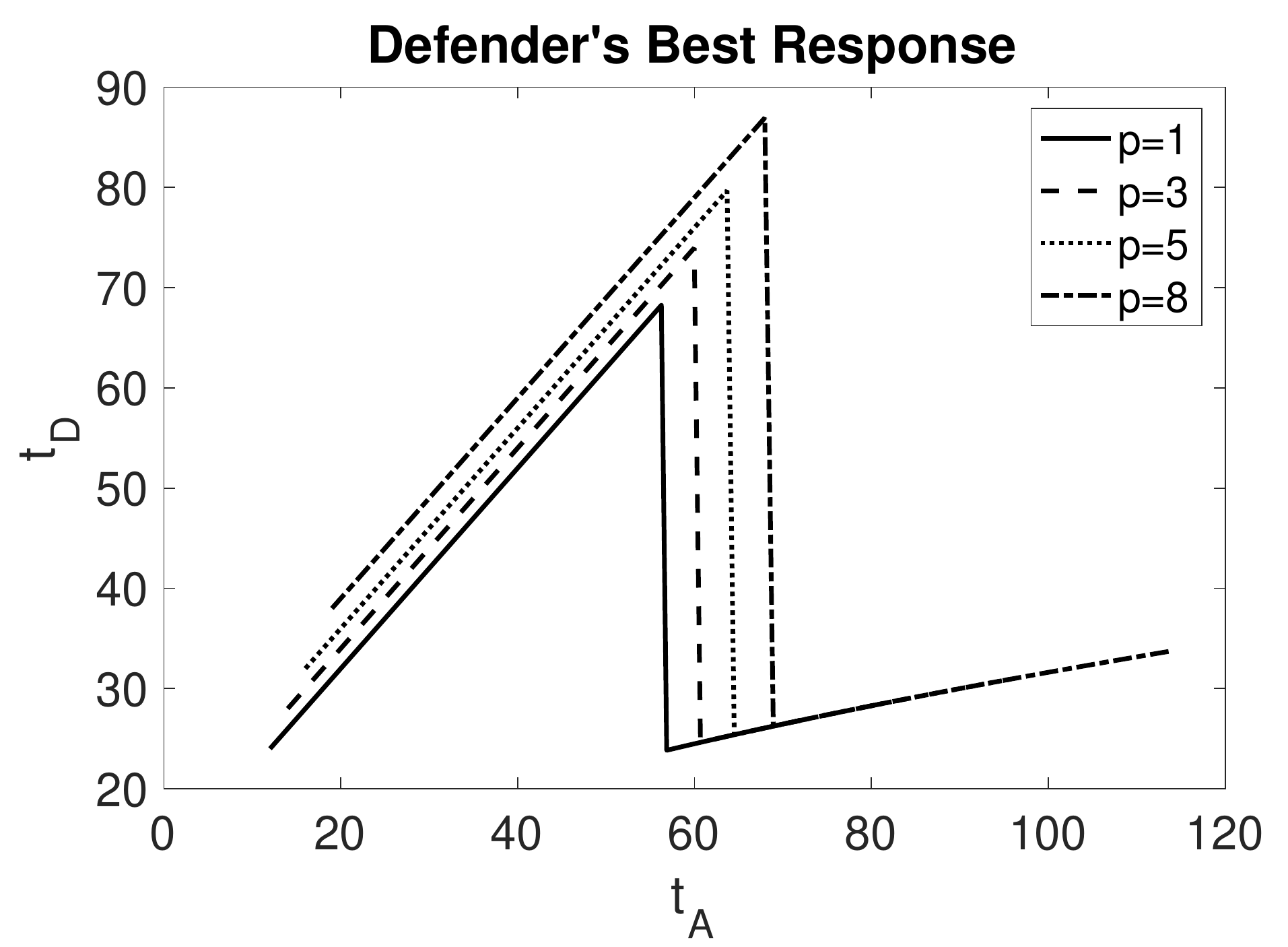}
			\label{fig:BrDd}}
	\end{subfigure}
	\begin{subfigure} [Attacker's best response]{ %
			\includegraphics[scale=0.35]{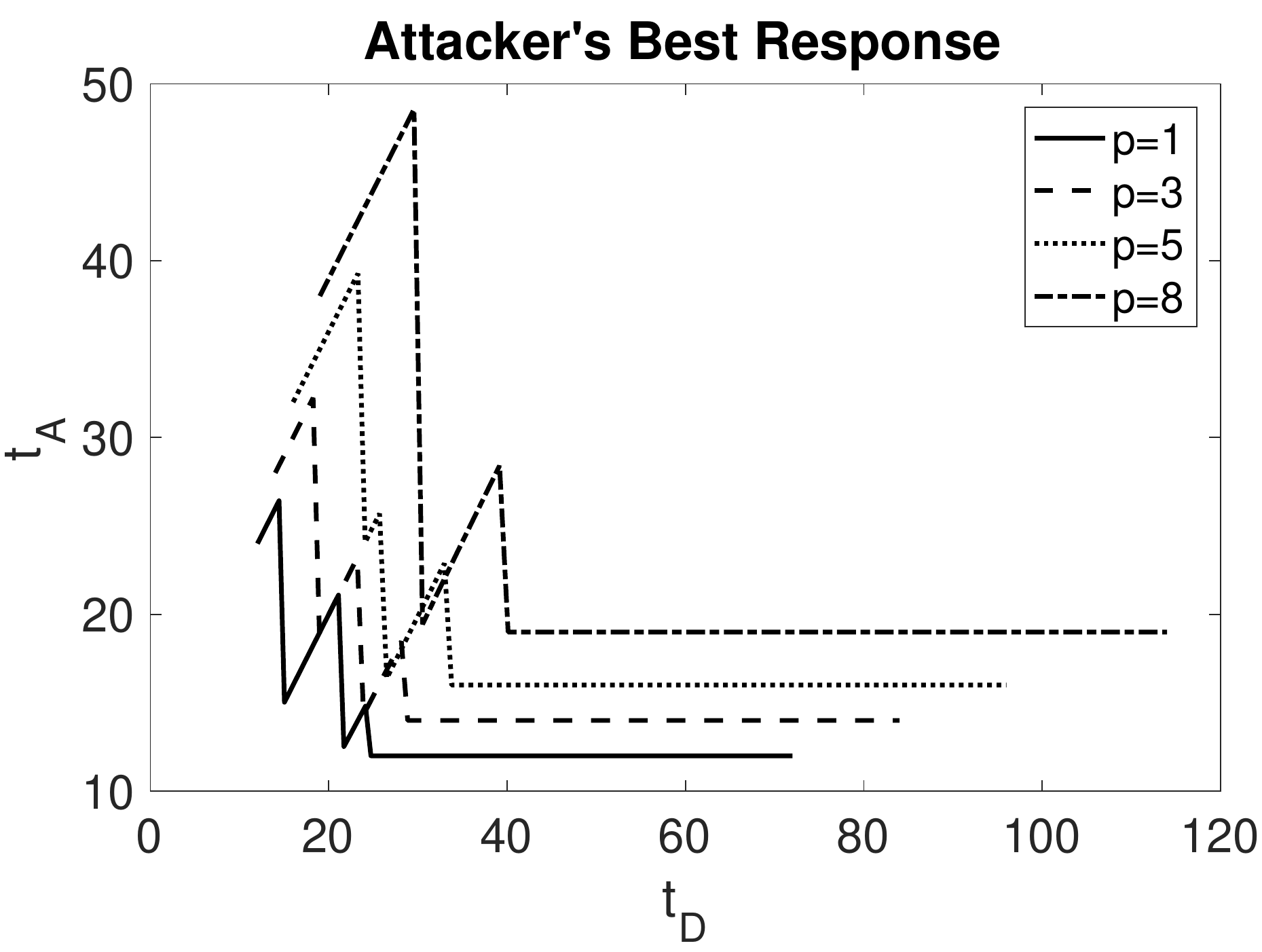} 			
			\label{fig:BrAd}}
	\end{subfigure}
	\caption{Players' best responses for different values of $\mathbf{p}$. We have $c_{\mathbf{D}}=10$, $c_\mathbf{k}=5$, $c_{\mathbf{A}} = 0.5$, $\mathbf{d}=10$, and $\mathbf{r}=1$. }
	\label{fig:BRd}
\end{figure}

Figures~\ref{fig:BrDd} and \ref{fig:BrAd} represent the defender's best response and the attacker's best response for different values of $\mathbf{p}$, respectively. Here, we have $c_{\mathbf{D}}=10$, $c_\mathbf{k}=5$, $c_{\mathbf{A}} = 0.5$, $\mathbf{d}=10$, and $\mathbf{r}=1$.  Similar to the previous figure, when the value of $\mathbf{p}$ increases, attacker and defender do not decrease $t_{\mathbf{A}}$ and $t_{\mathbf{D}}$, respectively. For the attacker, the higher the value of the protection time, the slower is the attacker's move pattern. The defender moves slower for higher values of protection time, if the attacker moves fast enough. Otherwise, the defender is indifferent, since the defender's best response is equal to $\sqrt{2c_\mathbf{k} t_{\mathbf{A}}}$, which is independent from $\mathbf{p}$, $\mathbf{d}$, and $\mathbf{r}$. 

\begin{figure}
	\centering
	\begin{subfigure} [Defender's best response]{ %
			\includegraphics[scale=0.35]{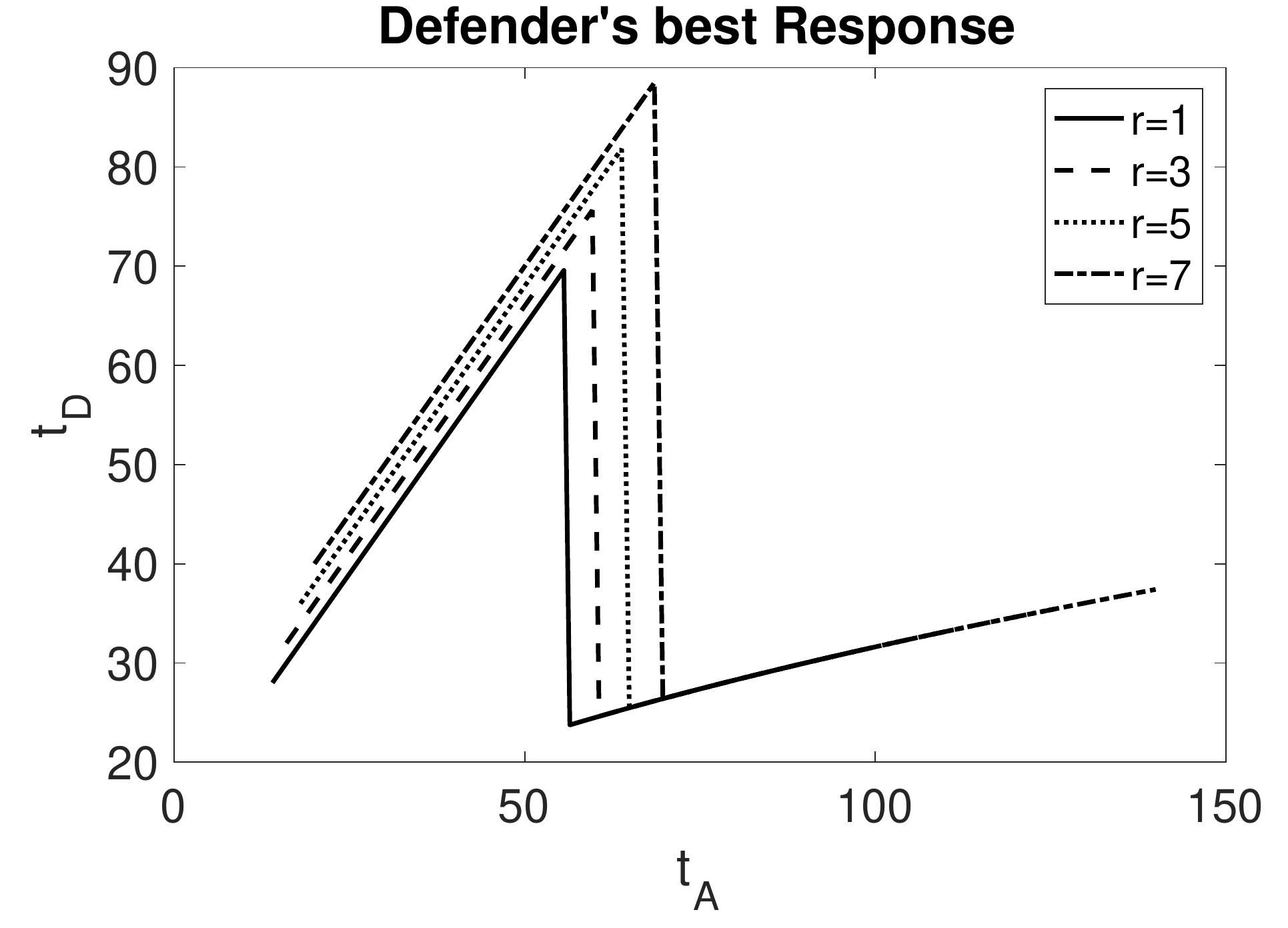}
			\label{fig:BrDr}}
	\end{subfigure}
	\begin{subfigure} [Attacker's best response]{ %
			\includegraphics[scale=0.35]{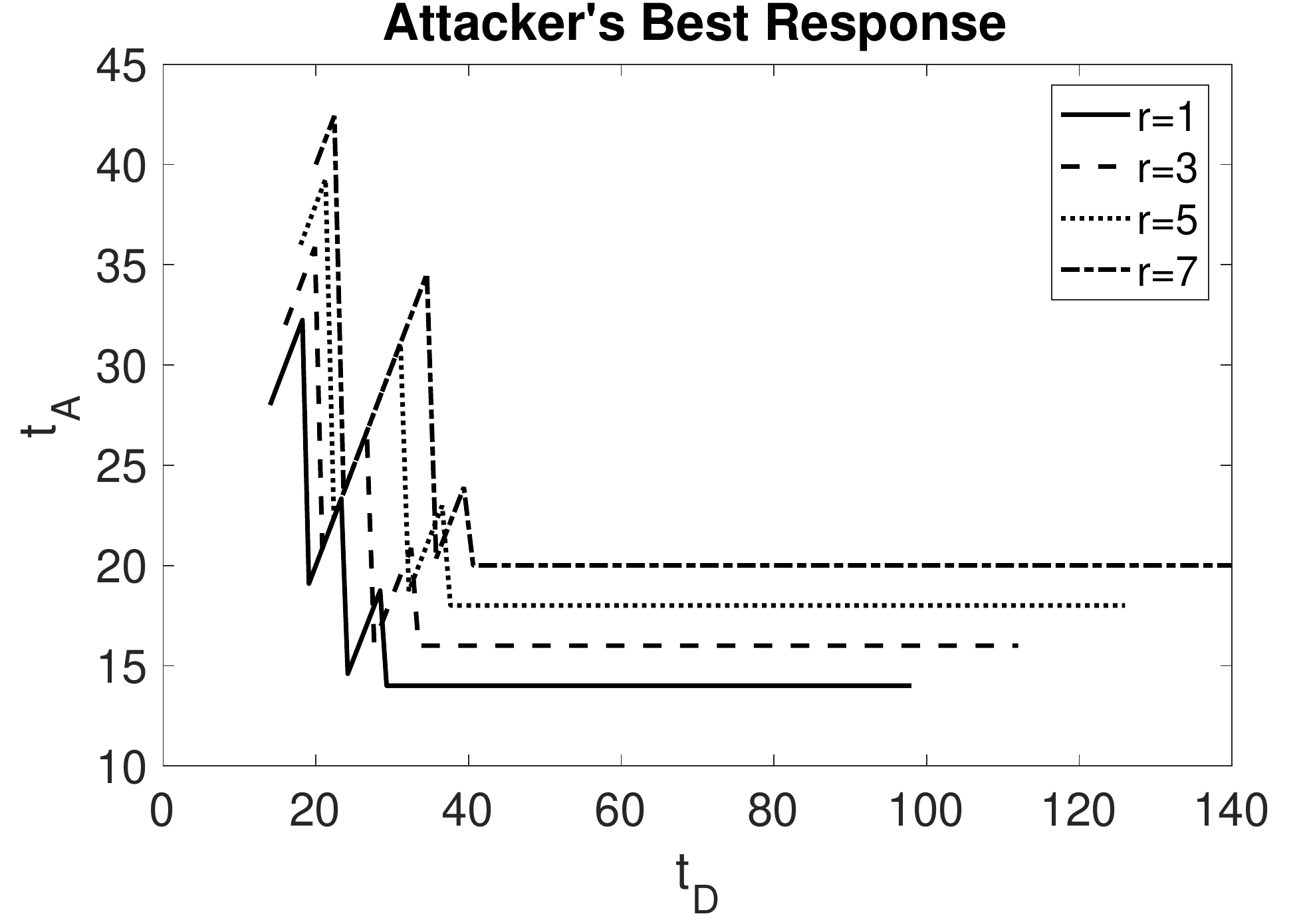} 			
			\label{fig:BrAr}}
	\end{subfigure}
	\caption{Players' best responses for different values of $\mathbf{r}$. We have $c_{\mathbf{D}}=2$, $c_\mathbf{k}=5$, $c_{\mathbf{A}} = 0.5$, $\mathbf{d}=10$, and $\mathbf{p}=3$. }
	\label{fig:BRr}
\end{figure}

Figure~\ref{fig:BRr} represents the effect of the reaction time on both players' best responses, which is similar to the two previous figures considering the effect of protection time and discovery time.  

Figure~\ref{fig:BrDcd} represents the role of $c_\mathbf{k}$ on the defender's best response. As we see in this figure, the defender's best response is equal to $t_{\mathbf{A}} + \mathbf{p+d+r}$ for small values of $t_{\mathbf{A}}$, which is independent from $c_\mathbf{k}$. For higher values of $c_\mathbf{k}$, the defender's best response is equal to $t_{\mathbf{A}}+\mathbf{p+d+r}$ for higher values of $t_{\mathbf{A}}$. Further, for high values of $t_{\mathbf{A}}$, the defender's best response is equal to $\mathbf{p+d+r}$ if $c_\mathbf{k}$ is low enough. Otherwise, the defender's best response is equal to $\sqrt{2c_\mathbf{k} t_{\mathbf{A}}}$.

In Figure~\ref{fig:BrDcD}, we consider the role of $c_{\mathbf{D}}$ on the defender's best response. As we can observe in this figure, the defender's best response switches between $t_{\mathbf{A}}+\mathbf{p+d+r}$ and $\sqrt{2t_{\mathbf{A}}c_\mathbf{k}}$ which are independent from $c_{\mathbf{D}}$. The only outcome depending on $c_{\mathbf{D}}$ can be found at the point where the defender switches from $t_{\mathbf{D}}+\mathbf{p+d+r}$ to $\sqrt{2t_{\mathbf{A}}c_\mathbf{k}}$. According to this figure, the defender switches its best response for higher values of $t_{\mathbf{A}}$ when $c_{\mathbf{D}}$ is higher. 

\begin{figure}
	\centering
	\begin{subfigure} [Defender's best response for different values of $c_\mathbf{k}$. We have $c_{\mathbf{D}}=10$]{ %
			\includegraphics[scale=0.35]{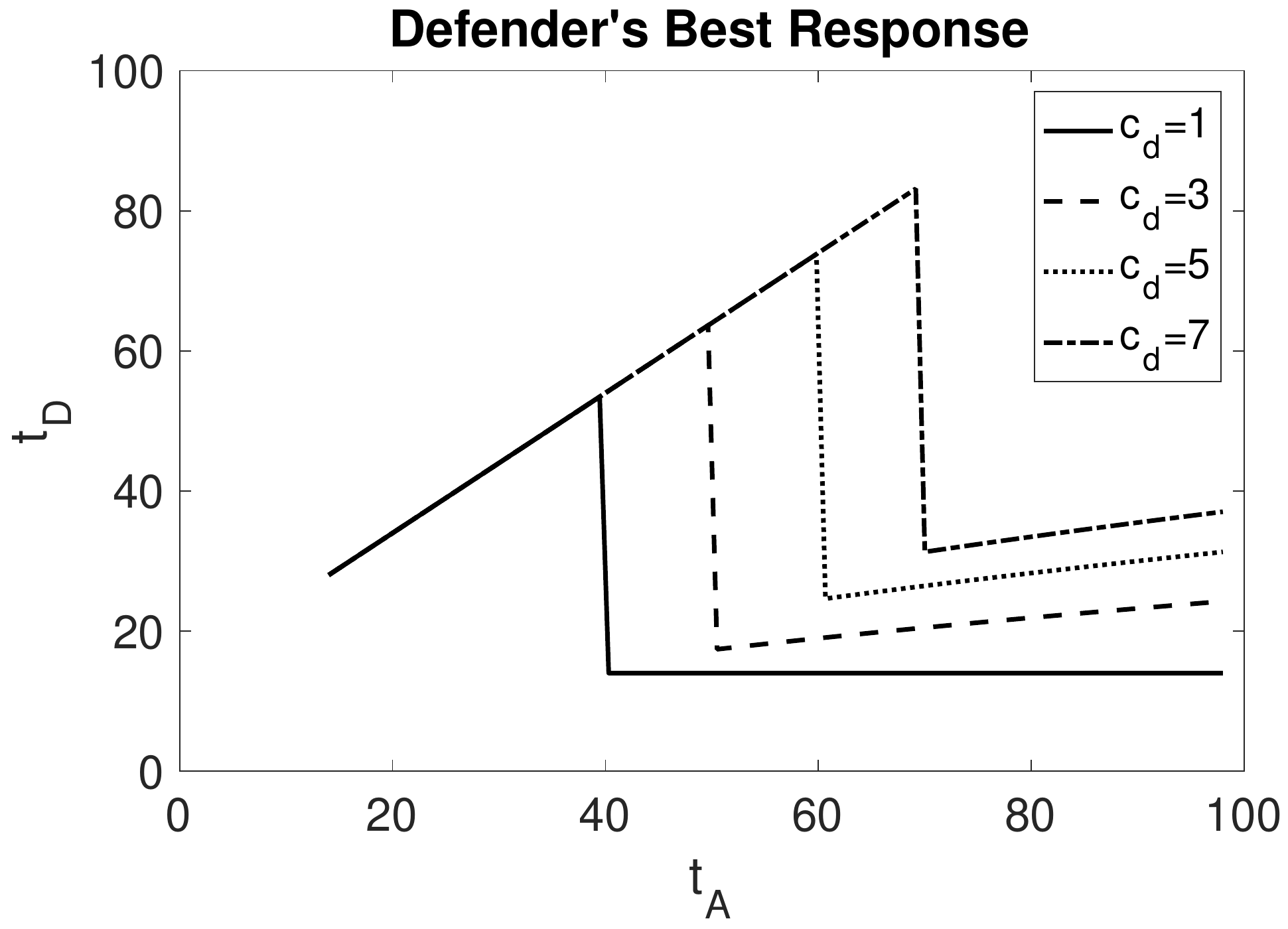}
			\label{fig:BrDcd}}
	\end{subfigure}
	\begin{subfigure} [Defender's best response for different values of $c_{\mathbf{D}}$. We have $c_\mathbf{k}=5$]{ %
			\includegraphics[scale=0.35]{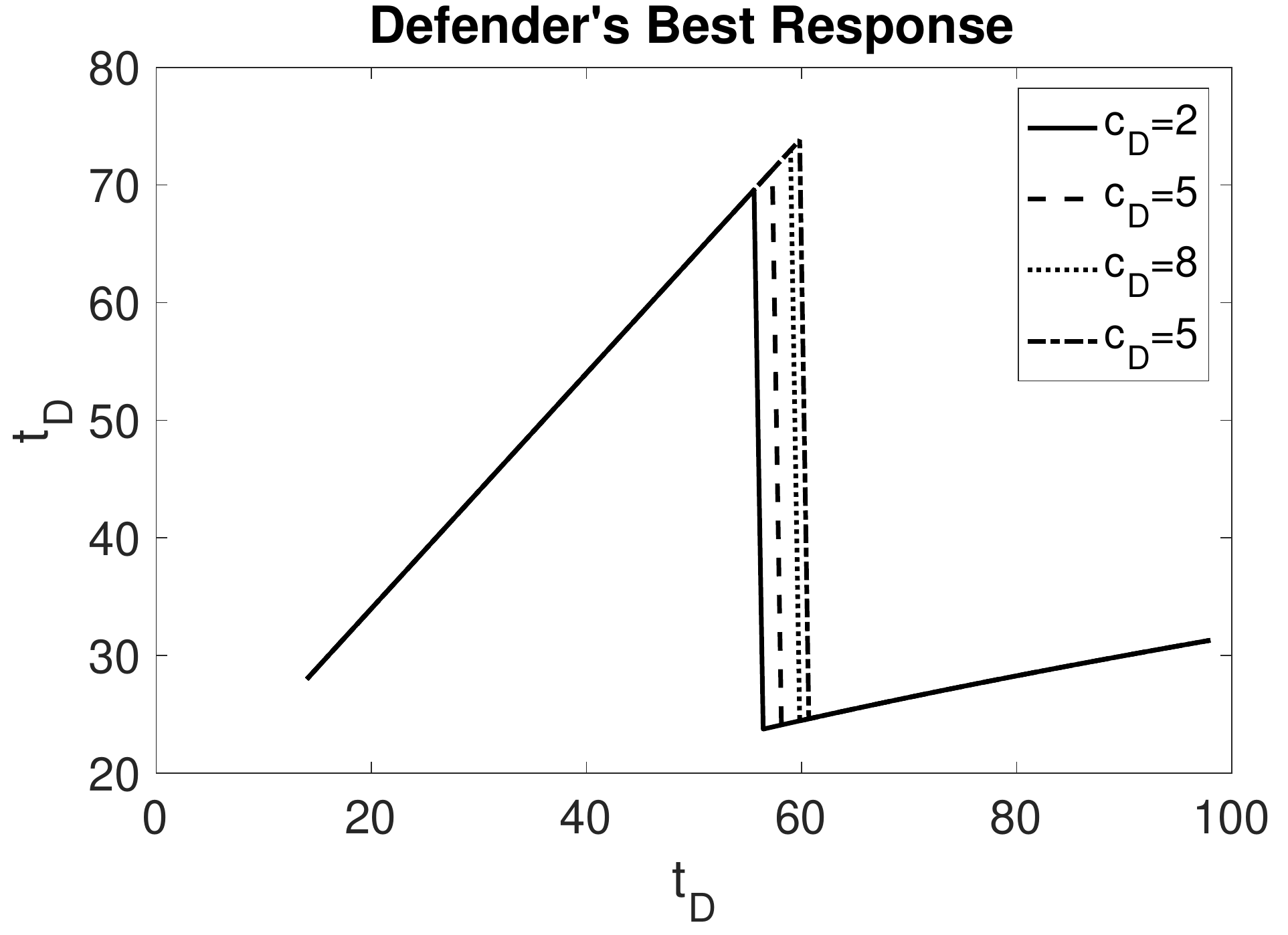} 			
			\label{fig:BrDcD}}
	\end{subfigure}
	\caption{Defender's best responses for different values of and $c_\mathbf{k}$ and $c_{\mathbf{D}}$. We have $c_{\mathbf{A}} = 0.5$, $\mathbf{d}=10$, $\mathbf{r}=1$, and $\mathbf{p}=3$. }
	\label{fig:BRcdcD}
\end{figure}

In Figure~\ref{fig:BRAttcA}, we investigate the role of $c_{\mathbf{A}}$ on the attacker's best response. When the defender moves fast, the attacker's best response is $t_{\mathbf{D}}+\mathbf{p+d+r}$ which is independent from the attacker's cost. For higher values of $t_{\mathbf{D}}$, the attacker's best response depends on the cost. In general, the higher the cost, the slower is the attacker's move.

In order to find any Nash equilibria, we calculate the best response of each player and then find the intersection, which is shown in Figure~\ref{fig:NE}. In this figure, we have $\mathbf{p}=3$, $\mathbf{d}=10$, $\mathbf{r}=1$, $c_\mathbf{k} = 5$, $c_{\mathbf{D}}=10$, and $c_{\mathbf{A}}=0.5$. According to this figure, the intersection of these two curves is at $(t_{\mathbf{A}},t_{\mathbf{D}})=(14.9,28.9)$. Note that the intersection of these two curves is at a discontinuity of the attacker's best response, i.e., the attacker's best response switches from one value to another (from one case to another case). However, the attacker's payoffs for these two best response strategies (at the discontinuity) are \textit{numerically} almost equal to each other, making them strategically equivalent. Note that this is similar to the original FlipIt game with periodic strategies where there exists an interval in which all points within the interval yield the same payoff and those points are part of the best response \cite{bowers2012defending,dijk2013flipit}. In this numerical example, $(t_{\mathbf{A}},t_{\mathbf{D}})=(14.9,28.9)$ is Nash equilibrium. 

\begin{figure}
	\centering
	\begin{subfigure} [Attacker's best response]{ %
			\includegraphics[scale=0.35]{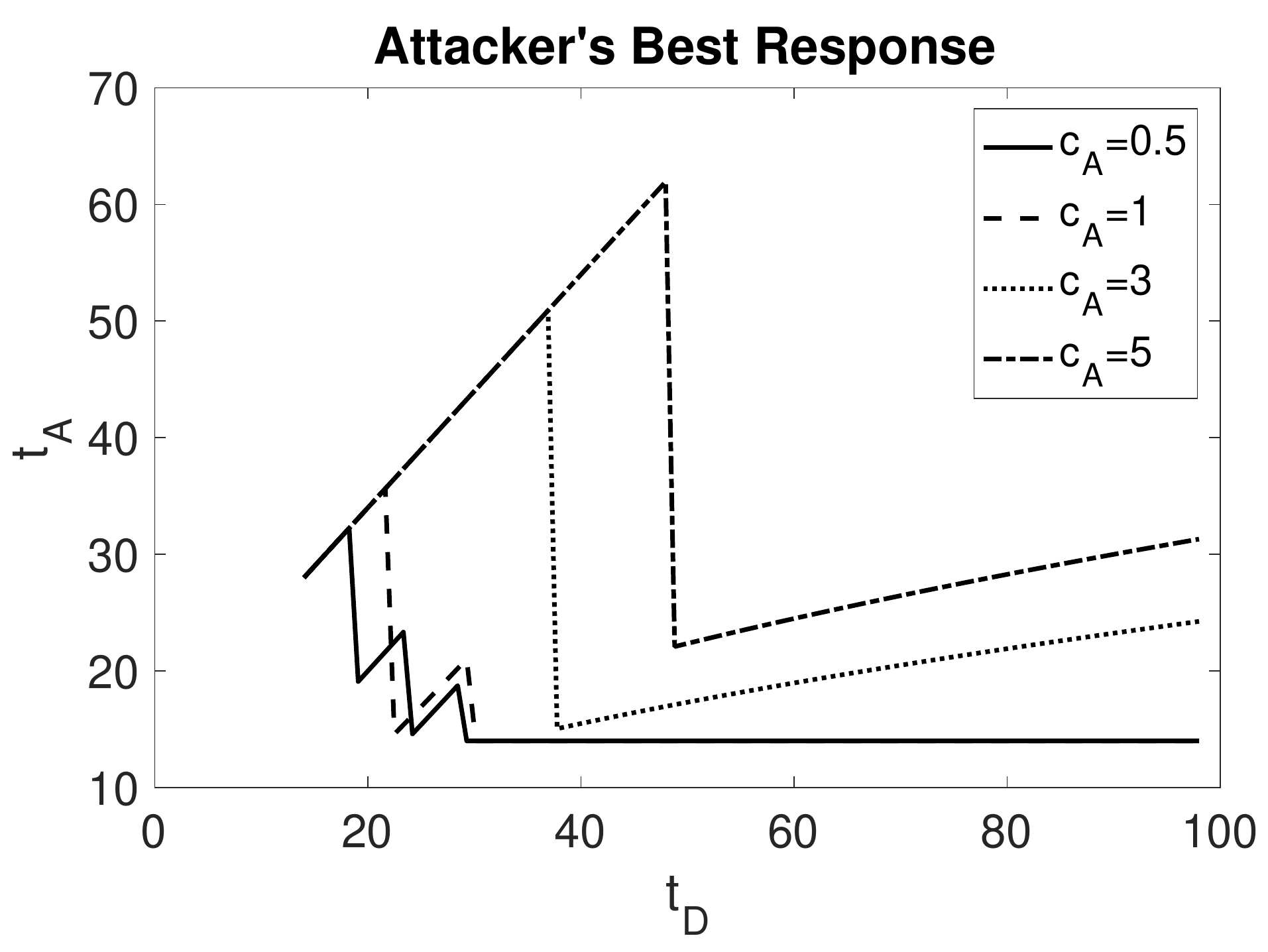}
			\label{fig:BRAttcA}}
	\end{subfigure}
	\begin{subfigure} [Nash Equilibrium]{ %
			\includegraphics[scale=0.35]{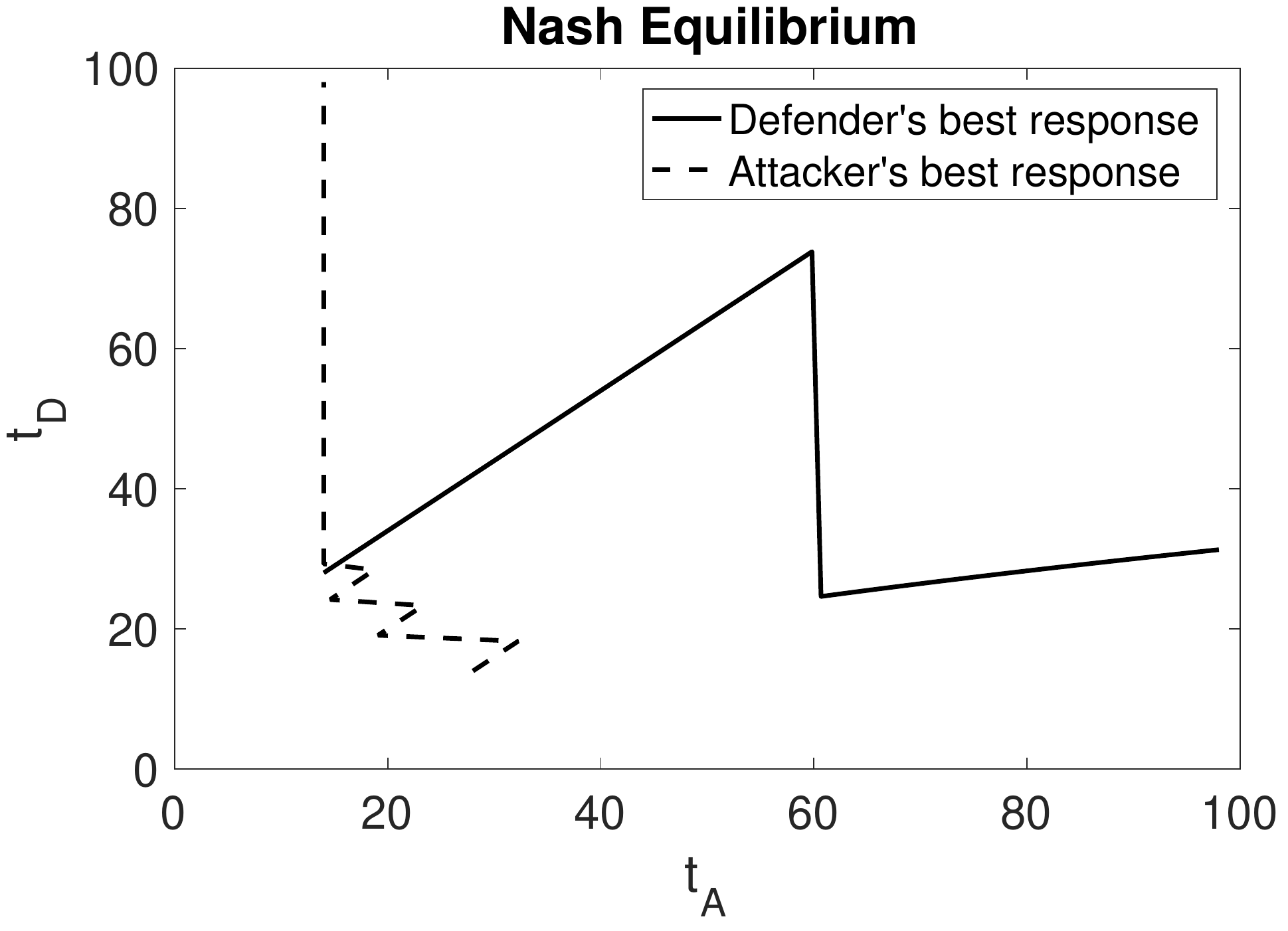} 			
			\label{fig:NE}}
	\end{subfigure}
	\caption{Attacker's best response for different values of $c_\mathbf{A}$ and Nash equilibrium }
	\label{fig:Com}
\end{figure}
%%%%%%%%%%%%%%%%%%%%%%%%%%%%%%%%% 
\section{Conclusion}
\label{sec:Conclu}
%\sadegh{Do we need to add a discussion section before conclusion?} 
In this paper, we first study the VCDB and screen other data sources to shed light on the question of the actual timing of security incidents and responses. We propose a distribution for the attack discovery time and provide heuristics about the distribution of the protection time and the reaction time in practice. While the gathered insights are useful, we assess the overall state of data collection for timing related data as severely lacking. The terminology for data collection is ambiguous (or at least seems to be interpreted unevenly by data contributors) and the collected data in the VCDB raises several questions. We are unaware of any superior data sources for a broad range of security issues.
 
Second, we propose a game-theoretic framework for Time-Based Security \cite{Schwartau99}. In particular, we aim to provide a richer framework to determine the defender's best time to reset the defense mechanism to a known safe state in the presence of a capable stealthy attacker. We incorporate the notions of protection time, detection time, and reaction time to provide a more realistic environment for the analysis of security scenarios unfolding over time.
%Moreover, we consider two variations of our model that differ regarding the defender's ability to detect attacks which are initially stealthy. 
%In our model, we calculate both players' payoff functions followed by analyzing our proposed game.
%which we then integrate in a comprehensive payoff function covering both cases. 
%The goal of this paper is to design an optimal defense update schedule policy which takes into account the defender who has the ability to detect attacks. 
%as well as the defender who is forced to make decisions while being ignorant about the presence of stealthy attack.

Next to the development of the payoff functions for both players, we analytically determine the defender's and the attacker's best responses. We evaluate our game with a numerical approach and do an example calculation for the corresponding Nash equilibrium of the game by visualizing both players' best responses and finding the intersection of these two plots. 
%Our findings are as follows: (1) It is not always the best choice for the defender to check the state of the resource continuously and to reset the resource to a safe state without delay. (2) Only for higher values of $c_{\mathbf{D}}$, $\mathbf{p}$, $\mathbf{d}$, and $\mathbf{r}$, the defender checks the state of the resource continuously when the attacker moves comparatively slowly. (3) Continuous checks of the state of the resource by the defender are not necessarily Nash equilibria of the game.

Our analysis is based on several assumptions and therefore provides meaningful opportunities for follow-up research. In particular, we study the case of the defender and the attacker acting periodically. While we observe periodic security behaviors in practice (such as fixed schedules for patch releases, or renewal of passwords and cryptographic keys), the study of other attacker and defender behaviors is equally well motivated. We anticipate the further development of Time-Based Security to positively impact the study of security games of timing, as well as security practice due to the increased relevance of the modeled scenario.

%we also conduct a preliminary numerical analysis for one of distinguishing cases. Our findings for this specific scenario are as follows: (1) The higher the defender's cost, the slower the defender's move rhythm. (2) The higher the protection time, the slower the defender's move rhythm. (3) The higher the detection probability, the slower the defender's move rhythm. For our future work, we aim to provide a thorough analysis for all possible cases and to extend both players' strategy spaces to renewal strategies to complement the analysis of the periodic strategy space.

\textbf{Acknowledgments:} We thank the reviewers for their detailed reviews. We further want to thank Aron Laszka for numerous suggestions for improvements based on an earlier version of this manuscript. Sadegh Farhang gratefully acknowledges a travel grant from the National Science Foundation to attend WEIS 2017. The research activities of Jens Grossklags are supported by the German Institute for Trust and Safety on the Internet (DIVSI). 

\bibliographystyle{amsplain}
\bibliography{gs1}

\providecommand{\bysame}{\leavevmode\hbox to3em{\hrulefill}\thinspace}
\providecommand{\MR}{\relax\ifhmode\unskip\space\fi MR }
\providecommand{\MRhref}[2]{%
  \href{http://www.ams.org/mathscinet-getitem?mr=#1}{#2}
}
\providecommand{\href}[2]{#2}
\begin{thebibliography}{10}

\bibitem{Agence14}
{Agence France Presse}, \emph{Swiss airforce grounded during hijacking because
  it was outside office hours}, 2014, (Last visited on May 28, 2017.) Available
  at:
  \url{http://www.huffingtonpost.com/2014/02/18/swiss-airforce-office-hours_n_4804151.html}.

\bibitem{axelrod2014timing}
R.~Axelrod and R.~Iliev, \emph{Timing of cyber conflict}, Proceedings of the
  National Academy of Sciences \textbf{111} (2014), no.~4, 1298--1303.

\bibitem{Blackwell49}
D.~Blackwell, \emph{The noisy duel, one bullet each, arbitrary accuracy}, Tech.
  report, The RAND Corporation, D-442, 1949.

\bibitem{bowers2012defending}
K.~Bowers, M.~{Van Dijk}, R.~Griffin, A.~Juels, A.~Oprea, R.~Rivest, and
  N.~Triandopoulos, \emph{Defending against the unknown enemy: {A}pplying
  {FlipIt} to system security}, Decision and Game Theory for Security,
  Springer, 2012, pp.~248--263.

\bibitem{Damballa}
{Damballa}, \emph{3\% to 5\% of enterprise assets are compromised by bot-driven
  targeted attack malware}, 2009, (Last visited on May 28, 2017.) Available at:
  \url{http://www.prnewswire.com/news-releases/3-to-5-of-enterprise-assets-are-compromised-by-bot-driven-targeted-attack-malware-61634867.html}.

\bibitem{edwards2016hype}
B.~Edwards, S.~Hofmeyr, and S.~Forrest, \emph{Hype and heavy tails: {A} closer
  look at data breaches}, Journal of Cybersecurity \textbf{2} (2016), no.~1,
  3--14.

\bibitem{farhang2016flip}
S.~Farhang and J.~Grossklags, \emph{Flipleakage: {A} game-theoretic approach to
  protect against stealthy attackers in the presence of information leakage},
  Decision and Game Theory for Security, Springer, 2016, pp.~195--214.

\bibitem{farhang2015phy}
S.~Farhang, Y.~Hayel, and Q.~Zhu, \emph{Phy-layer location privacy-preserving
  access point selection mechanism in next-generation wireless networks},
  Communications and Network Security (CNS), 2015 IEEE Conference on, IEEE,
  2015, pp.~263--271.

\bibitem{farhang2014dynamic}
S.~Farhang, H.~Manshaei, M.~Esfahani, and Q.~Zhu, \emph{A dynamic bayesian
  security game framework for strategic defense mechanism design}, Decision and
  Game Theory for Security, Springer, 2014, pp.~319--328.

\bibitem{feng_stealthy}
X.~Feng, Z.~Zheng, P.~Hu, D.~Cansever, and P.~Mohapatra, \emph{Stealthy attacks
  meets insider threats: {A} three-player game model}, Proceedings of MILCOM,
  2015.

\bibitem{grossklags2008secure}
J.~Grossklags, N.~Christin, and J.~Chuang, \emph{Secure or insure?: {A}
  game-theoretic analysis of information security games}, Proceedings of the
  17th International World Wide Web Conference, 2008, pp.~209--218.

\bibitem{grossklags2014task}
J.~Grossklags and D.~Reitter, \emph{How task familiarity and cognitive
  predispositions impact behavior in a security game of timing}, Proceedings of
  the 27th IEEE Computer Security Foundations Symposium (CSF), 2014,
  pp.~111--122.

\bibitem{hu2015dynamic}
P.~Hu, H.~Li, H.~Fu, D.~Cansever, and P.~Mohapatra, \emph{Dynamic defense
  strategy against advanced persistent threat with insiders}, Proceedings of
  the 34th IEEE International Conference on Computer Communications (INFOCOM),
  2015.

\bibitem{Johnson11}
B.~Johnson, R.~B{\"o}hme, and J.~Grossklags, \emph{Security games with market
  insurance}, Decision and Game Theory for Security, Springer, 2011,
  pp.~117--130.

\bibitem{johnson2015games}
B.~Johnson, A.~Laszka, and J.~Grossklags, \emph{Games of timing for security in
  dynamic environments}, Decision and Game Theory for Security, Springer, 2015,
  pp.~57--73.

\bibitem{kuypers2016empirical}
M.~Kuypers, T.~Maillart, and E.~Pat{\'{e}}-Cornell, \emph{An empirical analysis
  of cyber security incidents at a large organization}, Tech. report,
  Department of Management Science and Engineering, Stanford University; School
  of Information, UC Berkeley, 2016, (Last accessed May 28, 2017.)
  \url{https://cisac.fsi.stanford.edu/sites/default/files/kuypersweis_v7.pdf}.

\bibitem{Laszka_survey}
A.~Laszka, M.~Felegyhazi, and L.~Buttyan, \emph{A survey of interdependent
  information security games}, ACM Computing Surveys \textbf{47} (2014), no.~2,
  23:1--23:38.

\bibitem{laszka2014flipthem}
A.~Laszka, G.~Horvath, M.~Felegyhazi, and L.~Butty{\'a}n, \emph{Flip{T}hem:
  Modeling targeted attacks with {F}lip{I}t for multiple resources}, Decision
  and Game Theory for Security, Springer, 2014, pp.~175--194.

\bibitem{laszka2013mitigating}
A.~Laszka, B.~Johnson, and J.~Grossklags, \emph{Mitigating covert compromises},
  Proceedings of the 9th Conference on Web and Internet Economics (WINE),
  Springer, 2013, pp.~319--332.

\bibitem{laszka2013mitigation}
A.~Laszka, B.~Johnson, and J.~Grossklags, \emph{Mitigation of targeted and
  non-targeted covert attacks as a timing game}, Decision and Game Theory for
  Security, Springer, 2013, pp.~175--191.

\bibitem{leslie2015threshold}
D.~Leslie, C.~Sherfield, and N.~Smart, \emph{Threshold {F}lip{T}hem: {W}hen the
  winner does not need to take all}, Decision and Game Theory for Security,
  Springer, 2015, pp.~74--92.

\bibitem{liu2006bayesian}
Y.~Liu, C.~Comaniciu, and H.~Man, \emph{A {B}ayesian game approach for
  intrusion detection in wireless ad hoc networks}, Proceedings of the Workshop
  on Game Theory for Communications and Networks, 2006.

\bibitem{liu2015cloudy}
Y.~Liu, A.~Sarabi, J.~Zhang, P.~Naghizadeh, M.~Karir, M.~Bailey, and M.~Liu,
  \emph{Cloudy with a chance of breach: Forecasting cyber security incidents.},
  USENIX Security, 2015, pp.~1009--1024.

\bibitem{Manshaei13}
H.~Manshaei, Q.~Zhu, T.~Alpcan, T.~Bac\c{s}ar, and J.-P. Hubaux, \emph{Game
  theory meets network security and privacy}, ACM Computing Surveys \textbf{45}
  (2013), no.~3, 25:1--25:39.

\bibitem{moore2007examining}
T.~Moore and R.~Clayton, \emph{Examining the impact of website take-down on
  phishing}, Proceedings of the {APWG} 2nd Annual e{C}rime Researchers Summit,
  2007, pp.~1--13.

\bibitem{Nadella15}
S.~Nadella, \emph{Enterprise security in a mobile-first, cloud-first world},
  2015, (Last visited on May 28, 2017.) Available at:
  \url{http://news.microsoft.com/security2015/}.

\bibitem{nochenson2013behavioral}
A.~Nochenson and J.~Grossklags, \emph{A behavioral investigation of the
  {FlipIt} game}, 12th Workshop on the Economics of Information Security
  (WEIS), 2013.

\bibitem{pal_security}
R.~Pal, X.~Huang, Y.~Zhang, S.~Natarajan, and P.~Hui, \emph{On security
  monitoring in {SDN}s: {A} strategic outlook}, Tech. report, University of
  Southern California.

\bibitem{pawlick2015flip}
J.~Pawlick, S.~Farhang, and Q.~Zhu, \emph{Flip the cloud: {C}yber-physical
  signaling games in the presence of advanced persistent threats}, Decision and
  Game Theory for Security, Springer, 2015, pp.~289--308.

\bibitem{pham2012compromised}
V.~Pham and C.~Cid, \emph{Are we compromised? {M}odelling security assessment
  games}, Decision and Game Theory for Security, Springer, 2012, pp.~234--247.

\bibitem{PrivacyRights}
{Privacy Rights Clearinghouse}, \emph{Chronology of data breaches}, (Last
  visited on May 28, 2017.) Available at:
  \url{https://www.privacyrights.org/data-breaches}.

\bibitem{pu2014economic}
Y.~Pu and J.~Grossklags, \emph{An economic model and simulation results of app
  adoption decisions on networks with interdependent privacy consequences},
  Decision and Game Theory for Security, Springer, 2014, pp.~246--265.

\bibitem{radzik1996results}
T.~Radzik, \emph{Results and problems in games of timing}, Lecture
  Notes-Monograph Series, Statistics, Probability and Game Theory: Papers in
  Honor of David Blackwell \textbf{30} (1996), 269--292.

\bibitem{reitter2013risk}
D.~Reitter, J.~Grossklags, and A.~Nochenson, \emph{Risk-seeking in a continuous
  game of timing}, Proceedings of the 13th International Conference on
  Cognitive Modeling (ICCM), 2013, pp.~397--403.

\bibitem{sarabi2016risky}
A.~Sarabi, P.~Naghizadeh, Y.~Liu, and M.~Liu, \emph{Risky business:
  Fine-grained data breach prediction using business profiles}, Journal of
  Cybersecurity \textbf{2} (2016), no.~1, 15--28.

\bibitem{Schwartau99}
W.~Schwartau, \emph{Time based security: {P}ractical and provable methods to
  protect enterprise and infrastructure}, Networks and Nation Interpact Press,
  1999.

\bibitem{shokri2012protecting}
R.~Shokri, G.~Theodorakopoulos, C.~Troncoso, J.-P. Hubaux, and J.-Y. Le~Boudec,
  \emph{Protecting location privacy: {O}ptimal strategy against localization
  attacks}, Proceedings of the 2012 ACM Conference on Computer and
  Communications Security, 2012, pp.~617--627.

\bibitem{WebHacking}
{The Web Application Security Consortium}, \emph{{The Web Hacking Incident
  Database (WHID)}}, (Last visited on May 28, 2017.) Available at:
  \url{http://projects.webappsec.org/w/page/13246995/Web-Hacking-Incident-Database}.

\bibitem{dijk2013flipit}
M.~Van~Dijk, A.~Juels, A.~Oprea, and R.~Rivest, \emph{Flipit: {T}he game of
  ``stealthy takeover''}, Journal of Cryptology \textbf{26} (2013), no.~4,
  655--713.

\bibitem{vcdb}
{VERIS Community Database}, \url{http://vcdb.org/}.

\bibitem{VERIS}
{VERIS Framework}, http://veriscommunity.net/.

\bibitem{DBIR}
{Verizon Enterprise}, \emph{{Data Breach Investigations Report (DBIR)}}, (Last
  visited on May 28, 2017.) Available at:
  \url{http://www.verizonenterprise.com/verizon-insights-lab/dbir/2016/}.

\bibitem{wellman2014empirical}
M.~Wellman and A.~Prakash, \emph{Empirical game-theoretic analysis of an
  adaptive cyber-defense scenario ({P}reliminary report)}, Decision and Game
  Theory for Security, Springer, 2014, pp.~43--58.

\bibitem{zhang_stealthy}
M.~Zhang, Z.~Zheng, and N.~Shroff, \emph{Stealthy attacks and observable
  defenses: {A} game theoretic model under strict resource constraints},
  Proceedings of the IEEE Global Conference on Signal and Inf. Processing
  (GlobalSIP), 2014, pp.~813--817.

\end{thebibliography}

%\singlespacing
%\setlength\bibsep{0pt}
%\bibliographystyle{my-style}
%\bibliography{Placeholder}

%\clearpage

%\onehalfspacing

%\section*{Tables} \label{sec:tab}
%\addcontentsline{toc}{section}{Tables}

%\clearpage

%\section*{Figures} \label{sec:fig}
%\addcontentsline{toc}{section}{Figures}

%\begin{figure}[hp]
%  \centering
%  \includegraphics[width=.6\textwidth]{../fig/placeholder.pdf}
%  \caption{Placeholder}
%  \label{fig:placeholder}
%\end{figure}

%\clearpage

%\section*{Appendix A. Placeholder} \label{sec:appendixa}
%\addcontentsline{toc}{section}{Appendix A}

\end{document}